\documentclass{article}
\usepackage{graphicx} % Required for inserting images
\usepackage{preamble}
\title{The $R$-Matrix in 3d Topological BF Theory}
\author{Nanna Havn Aamand}

\date{\vspace{-1em}
    \begin{small}
        \textit{University of Copenhagen,
         Department of Mathematical Sciences, \\
        2100 Copenhagen, Denmark}\\
        [3ex]
    \end{small} 
}

\begin{document}

\maketitle

\begin{abstract}
    In this paper I study Wilson line operators in a certain type of ``split'' Chern-Simons theory for a Lie algebra $\g=\mathfrak a\oplus \mathfrak a^*$ on a manifold with boundaries. The resulting gauge theory
    %for which the Lie algebra decomposes into a direct sum of maximal isotropic subalgebras. One motivation for studying this type of Chern-Simons theory is it 
    is a 3d topological BF theory equivalent to a topologically twisted 3d $\mathcal N=4$ theory. I show that this theory realises solutions to the quantum Yang-Baxter equation all orders in perturbation theory as the expectation value of crossing Wilson lines.
\end{abstract}
\newpage
\tableofcontents
\newpage
\section{Introduction}
The perturbative framework for Chern-Simons theory on a general three-manifold $M$ was formalised by Axelrod and Singer in \cite{AS}. To account for ultraviolet singularities in Feynman integrals they used a Fulton-MacPherson like compactification of the configuration space of Feynman diagram vertices in $M$. The compactified space has the form of a stratified space with boundary strata defined from spherical blow-ups along the diagonals where subsets of vertices come together. This has led to a technique for recovering manifold invariants from Chern-Simons theory implemented in a series of notable works, see e.g. \cite{K,BT,BC,AF}. In particular, Bott and Taubes \cite{BT} constructed knot invariants from Wilson loops in $S^3$. The essential ingredient in this work is the use of Stokes' theorem: Since propagators in the theory are closed forms, proving invariance of the expectation value of Wilson loops under continuously displacing loop strands amounts to proving a series of vanishing theorems for Feynman integrals on the boundary of the configuration space. 
%suggested by Kontsevich \cite{K} and implemented in among others notable works of Bott and Taubes \cite{BT}, Bott and Cattaneo  \cite{BC} and Altschuler and Freidel \cite{AF}.
The objective of this paper is to implement the same type of arguments for the purpose of recovering a solution to the Yang-Baxter equation (an $R$-matrix) from the expectation value of crossing Wilson lines at all orders in perturbation theory. 

%However, things become much simpler if we instead
In \cite{Aamand} the present author carried out leading order Feynman diagram computations to realise the classical Yang-Baxter equation from Wilson lines in Chern-Simons theory for a semi-simple Lie algebra $\g$, on a manifold with boundaries $\R^2\times [-1,1]$. In order to obtain Yang-Baxter solutions, one must place boundary condition on the gauge field to break the full gauge symmetry of the theory. This is achieved by extending the Lie algebra by an extra copy of the Cartan subalgebra to admit a decomposition into maximal isotropic subalgebras $\g=\mathfrak l_-\oplus \mathfrak l_+$, restricting the gauge field to $\mathfrak l_-$ (resp. $\mathfrak l_+$) on the upper (resp. lower) boundary. This work was inspired by a construction of Costello, Witten and Yamazaki \cite{CWYI}, \cite{CWYII} in a 4-dimensional analogue of Chern-Simons theory. In this framework, the Yang-Baxter equation states the equivalence between the diagrams on the left- and right-hand side of figure \ref{fig:1}, where the lines represent Wilson lines extending to infinity along $\R^2$ and supported at different points in $[-1,1]$. The corresponding expectation value is an element in $\mathcal U(\g)^{\otimes 3}[[\hbar]]$.
\begin{figure}[H]
    \centering
\begin{tikzpicture}[vertex/.style={draw,circle, fill=black, inner sep=1.5pt}]

  \node (a1) at (-0.7,1.7)  {}; 
  \node (a2) at (-0.7,-0.65)  {};
  \node (a3) at (-0.7,-0.85)  {}; 
  \node (a4) at (-0.7,-1.7)  {};
  \draw[very thick] (a1)node[above]{$L_{2,0}$} --  (a2);
  \draw[very thick,->] (a3) --  (a4);

  \node (b1) at (-1.7,1.7)  {}; 
  \node (b2) at (-0.75,0.75)  {};
   \node (b3) at (-0.65,0.65)  {}; 
  \node (b4) at (0,0)  {};
  \node (b5) at (0.05,-0.05)  {};
  \node (b6) at (1.7,-1.7)  {};
  \draw[very thick] (b1) node[above]{$L_{1,0}$} -- (b2);
  \draw[very thick] (b3) -- (b4);
  \draw[very thick,->] (b5) -- (b6);
  
  \node (c1) at (1.7,1.7)  {}; 
  \node (c2) at (-1.7,-1.7)  {};
  \draw[very thick,->] (c1)node[above]{$L_{3,0}$} --  (c2);

\end{tikzpicture} \hspace{2.5cm}
\begin{tikzpicture}[vertex/.style={draw,circle, fill=black, inner sep=1.5pt}]

  \node (a1) at (0.7,1.7)  {};
  \node (a2) at (0.7,0.8)  {};
  \node (a3) at (0.7,0.6)  {};
  \node (a4) at (0.7,-1.7)  {};
  \draw[very thick] (a1)node[above]{$L_{2,1}$} --  (a2);
  \draw[very thick,->] (a3) --  (a4);

  \node (b1) at (-1.7,1.7)  {}; 
  \node (b2) at (0,0)  {};
  \node (b3) at (0.05,-0.05)  {};
  \node (b4) at (0.65,-0.65)  {};
  \node (b5) at (0.75,-0.75)  {};
  \node (b6) at (1.7,-1.7)  {};
  \draw[very thick] (b1) node[above]{$L_{1,1}$} -- (b2);
  \draw[very thick] (b3) -- (b4);
  \draw[very thick,->] (b5) -- (b6);
  
  \node (c1) at (1.7,1.7)  {}; 
  \node (c2) at (-1.7,-1.7)  {};
  \draw[very thick,->] (c1)node[above]{$L_{3,1}$} --  (c2);

\end{tikzpicture} 
\caption{: The Yang-Baxter equation for crossing Wilson lines.}
    \label{fig:1}
\end{figure}
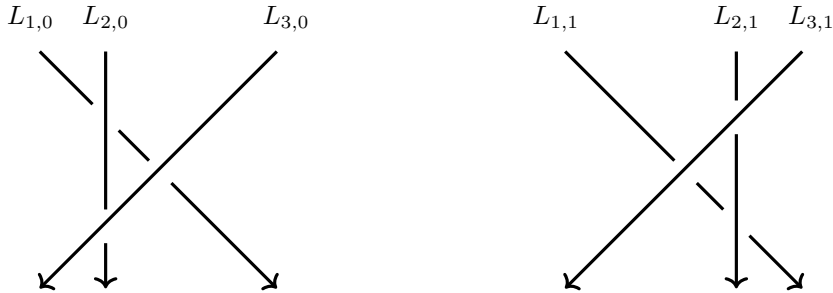
\noindent Directly implementing vanishing arguments similar to those of Bott and Taubes to the above theory
%in order to prove invariance under continuously displacing the middle Wilson line. This however
appears too ambitious as the vanishing theorems rely on a full rotational symmetry of the propagator which in this case is broken by the boundary conditions. However, things become easier if we instead consider %Chern-Simons theory in the same setup but for 
a Lie algebra $\g=\mathfrak a\oplus \mathfrak a^*$ with relations $[a,b^*]=[a,b]^*$ and $[a^*,b^*]=0$ for $a,b\in \mathfrak a$. Chern-Simons theory for this Lie algebra is equivalent to a $B$-twisted $3$d $\mathcal N=4$ theory; see e.g.
%The above gauge theory has recently been subject to an increased interest since it is equivalent to a topologically twisted $3$d $\mathcal N=4$ gauge theory; see e.g.
\cite{CG, Garner, G2023}. For this theory Feynman diagrams become particularly simple. In fact, the gauge field decomposes into two parts $\mathbf A\in \Omega^1(M)\otimes \mathfrak a$ and $\mathbf B\in \Omega^1(M)\otimes \mathfrak a^*$, and the only type of interaction vertices permitted by the theory has one incoming $\mathbf B$-field and two outgoing $\mathbf A$-fields. It turns out that this accounts for the problematic boundary faces 
%and the vanishing arguments carry through with small modifications. 
and we can therefore prove the following theorem:
%, implying that the expectation value of a pair of crossing Wilson lines is a solution to the Yang-Baxter equation
\begin{theorem}\label{thm:1}
Let $\braket{L_t}$ be the expectation value of the product of Wilson lines in figure \ref{fig:1}, where the parameter $t$ corresponds to moving the middle line continuously to the right. In the theory described above it holds that $\braket{L_1}-\braket{L_0}=0$.
\end{theorem}
This entails proving a series of vanishing theorems in line with those of Bott and Taubes. The perturbative formalism for this ``split'' Chern-Simons theory on a manifold with boundaries was first studied in work of Cattaneo et al. \cite{CMR,CMK}, from where the term originates. 

\section{The Quantum Yang-Baxter Equation} \label{sec:YBE}
We begin by briefly recalling some basic notions relating to the quantum Yang-Baxter equation. Let $\g$ be a Lie algebra that can be quantized via the Drinfel'd double construction and let $\mathcal U_\hbar(\g)$ be the corresponding quantized universal enveloping algebra of $\g$. For each $i,j\in\{1,2,3\}$ with $i\neq j$ define 
$
\rho_{ij}:\mathcal U_\hbar(\g)^{\otimes 2}\to \mathcal U_\hbar(\g)^{\otimes 3}
$ 
by
\begin{align*}
    \rho_{12}(a\otimes b)=a\otimes b\otimes 1 \, , \ \rho_{13}(a\otimes b)=a\otimes 1\otimes b \, ,  \ \rho_{23}(a\otimes b)=1\otimes a\otimes b
\end{align*}
%It holds that $\mathcal U_\hbar(\g)$ has the structure of a quasi-triangular Hopf algebra with braiding given by an element 
%It holds that
%{\color{red} in the sense of e.g. Drinfel'd}. 
%$\mathcal U_\hbar(\g)$ has the structure of a quasi-triangular Hopf algebra with the braiding coming from an element 
%Since we are working in perturbation theory we write represent the $R$-matrix as 
%which we represent by its perturbative expansion in the parameter~$\hbar$: 
Given an element $R_\hbar\in \mathcal U_\hbar(\g)\otimes \mathcal U_\hbar(\g)$, write $R_{ij} = \rho_{ij}(R_\hbar)$.
%we can write
%$$
%R_\hbar=1\otimes 1+\sum_{k=1}^\infty\hbar^{k}r^{(k)},
%$$ 
%where $r^{(k)}\in \mathcal U(\g)\otimes \mathcal U(\g)$ for each $k$. We 
We say that $R_\hbar$ is a quantum $R$-matrix if it is invertible and it satisfies the following relation known as the Yang-Baxter equation: 
\begin{equation}\label{eq:YBE}   R_{23}R_{13}R_{12}=R_{12}R_{13}R_{23}\,,
\end{equation}
This equation is commonly represented graphically by the diagram shown below. 
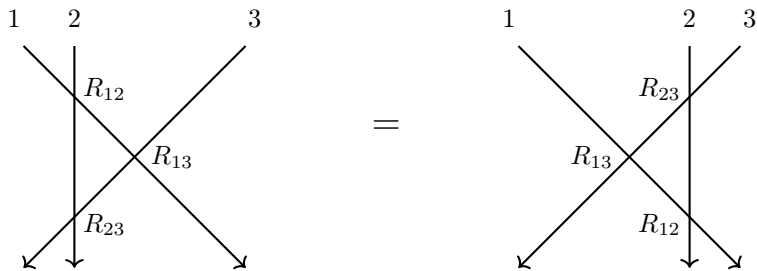
\begin{figure}[H]
    \centering
\begin{tikzpicture}
  \node (a1) at (-0.8,1.6)  {}; 
  \node (a2) at (-0.8,-1.6)  {};
  \draw[thick, ->] (a1)node[above]{2} --  (a2);

  \node (b1) at (-1.6,1.6)  {}; 
  \node (b2) at (1.6,-1.6)  {};
  \draw[thick,->] (b1)node[above]{1} -- (b2);
  
  \node (c1) at (1.6,1.6)  {}; 
  \node (c2) at (-1.6,-1.6)  {};
  \draw[thick,->] (c1)node[above]{3} --  (c2);

  \node at (-0.4,0.9) {$R_{12}$};
  \node at (0.5,0) {$R_{13}$};
  \node at (-0.4,-0.9) {$R_{23}$};
\end{tikzpicture} \hspace{1cm}
\begin{tikzpicture}
  \node (a1) at (0,0)  {};
  \node (a2) at (0,2)  {\Large $=$};
\end{tikzpicture}\hspace{1cm}
\begin{tikzpicture}[vertex/.style={draw,circle, fill=black, inner sep=1.5pt}]
\node (a1) at (0.8,1.6)  {}; 
  \node (a2) at (0.8,-1.6)  {};
  \draw[thick, ->] (a1)node[above]{2} --  (a2);

  \node (b1) at (-1.6,1.6)  {}; 
  \node (b2) at (1.6,-1.6)  {};
  \draw[thick,->] (b1)node[above]{1} -- (b2);
  
  \node (c1) at (1.6,1.6)  {}; 
  \node (c2) at (-1.6,-1.6)  {};
  \draw[thick,->] (c1)node[above]{3} --  (c2);

  \node at (0.4,0.9) {$R_{23}$};
  \node at (-0.5,0) {$R_{13}$};
  \node at (0.4,-0.9) {$R_{12}$};
\end{tikzpicture}
%\caption{: The embedding space $L_t$ where $t\in [0,1]$.}
%    \label{fig:W-lines}
%\end{figure}
%\begin{figure}[H]
%    \centering
%    \includegraphics[scale=0.25]{YBE1.png}
    %\caption{:~A graphical representation of the Yang-Baxter equation}
    %\label{fig:YBE}
\end{figure}
\noindent To interpret this diagram, we imagine that each line carries a vector space $V_i$, $i\in\{1,2,3\}$ corresponding to some representation of $\g$. At the crossing between line $i$ and line $j$ the incoming vector spaces are transformed by the element $R_{ij}\in \End(V_i\otimes V_j)$ acting in the given representation.Reading the figure from up to down in the direction of the arrow reproduces the Yang-Baxter equation. The existence of an $R$-matrix gives a braiding structure on $\mathcal U_\hbar(\g)$, and hence in particular it allows for the construction of invariants of knots and braids.
\section{Split Chern-Simons Theory with Boundaries}\label{sec:gauge theory}
\subsection{The basic setup}
Let $\g$ be a Lie algebra with a non-degenerate invariant pairing $\Tr:\g\otimes \g\to \mathbb \R$ and assume that $\g$ admits a decomposition $\g=\mathfrak a\oplus \mathfrak a^*$ %be a Lie algebra with an 
where $\mathfrak a^*$ is dual to $\mathfrak a$ with respect $\Tr$. Moreover, let $\mathcal B(\mathfrak a)=\{\xi^a\}_{a=1,\dots, \dim \mathfrak a}$ be a basis for $\mathfrak a$ and $\mathcal B(\mathfrak a^*)=\{\zeta_a\}_{a=1,\dots, \dim \mathfrak a}$ be the dual basis for $\mathfrak a^*$. The gauge theory that we study in this paper is Chern-Simons for the Lie algebra $\g$ described above, with relations
%and assume that $\g$ admits a decomposition $\g=\mathfrak a\oplus \mathfrak a^*$ where $\mathfrak a^*$ is dual to $\mathfrak a$ with respect to the pairing. 
\begin{equation}\label{eq:brackets}
[\xi_a,\xi_b]={f^c}_{ab}\xi_c \ , \ \ [\zeta^a,\xi_b]={f^a}_{bc}\zeta^c \ , \ \ [\zeta^a,\zeta^b]=0,
\end{equation}
where ${f^c}_{ab}$ are the structure constants of $\mathfrak a$. Notice that, with this definition, $\mathfrak a$ and $\mathfrak a^*$ are maximal isotropic subalgebras of $\g$ and hence the triple $(\g,\mathfrak a,\mathfrak a^*)$ is a Manin triple. This is, in essence, what allows us to derive quantum groups structured in the theory. The above gauge theory is defined by the Chern-Simons action:
\begin{equation}\label{eq:action}
S_\text{CS}(\mathbf{C})=\frac{1}{2\pi}\int_{M}\Tr(\mathbf{C}\wedge d\mathbf{C})+\frac{1}{3}\Tr([\mathbf{C}, \mathbf{C}]\wedge\mathbf{C}),
\end{equation}
where the gauge field $\mathbf C$ is a one-form on a manifold $M$ taking values in $\g$, i.e. $\mathbf C\in\Omega^1(M)\otimes \g$. We will decompose $\mathbf C$ into a part $\mathbf A$ taking value in $\mathfrak a$ and a part $\mathbf B$ taking value in $\mathfrak a^*$. That is, we write 
$$
\mathbf{C}=\mathbf A + \mathbf B,
$$
where $\mathbf A \in \Omega^1(M)\otimes \mathfrak a$ and $\mathbf B\in \Omega^1(M)\otimes \mathfrak a^*$. Observe that, when inserting this into the Chern-Simons action, the terms containing only $\mathbf A$'s or $\mathbf B$'s vanish since the subalgebras $\mathfrak a$ and $\mathfrak a^*$ are isotropic. Similarly the term $\Tr({[\mathbf{A}, \mathbf{B}],\mathbf{B}})$ vanish by the relations in equation \eqref{eq:brackets}. Thus the resulting action takes the form:   
\begin{equation}
\begin{aligned}
S_\text{CS}(\mathbf{A}+\mathbf{B})=\frac{1}{2\pi}\int_{M}\Tr({\mathbf{A}\wedge d\mathbf{B}}+{\mathbf{B}\wedge d\mathbf{A}})+\frac{1}{3}\Tr({[\mathbf{A}, \mathbf{A}]\wedge\mathbf{B}}),
%\\
%&=\int_{\R^2\times I}\braket{\mathbf{A}\wedge d\mathbf{B}}+\frac{1}{6}\braket{[\mathbf{A}, \mathbf{A}]\wedge\mathbf{B}}+\int_{\R^2\times \{-1,1\}}\braket{\mathbf A\wedge \mathbf B}\,.
\end{aligned}
\end{equation}
which we identify with the action of a 3d topological BF theory. The first term in the above action is a kinetic term and represents the free propagation of a gauge field between states $\mathbf A$ and $\mathbf B$. We use a convention where the corresponding propagator is represented by an oriented edge going from $\mathbf A$ to $\mathbf B$. The form of the cubic interaction term then implies that the only allowed interaction vertices in the theory are the of the form shown in figure \ref{fig:interaction}, with one incoming $\mathbf B$-edge and two outgoing $\mathbf{A}$-edges. We will say more on this in section \ref{sec:admissible}.
\begin{figure}[H]
    \centering
   \begin{tikzpicture}[vertex/.style={draw,circle, fill=black, inner sep=1.5pt}]
  \node[vertex] (v0) at (0,0)  {}; 
  \node (v1) at (0,-1)  {};
  \node (v2) at (-0.8,0.7) {};
  \node (v3) at (0.8,0.7) {};
  \draw[<-] (v0) -- (v1) node[below]{$\mathbf B$};  
  \draw[->] (v0) -- (v2) node[above]{$\mathbf A$};
  \draw[->] (v0) --  (v3)node[above]{$\mathbf A$};

  %\draw[->] (0,0.4) arc (90:330:0.4);
\end{tikzpicture}
    \caption{: Interaction vertices in the relevant split Chern-Simons theory}
    \label{fig:interaction}
\end{figure}
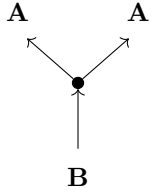
\noindent In what follows we take $M$ to be a manifold with boundaries, $M=\R^2\times I$, where $I=[-1,1]$. In this setting, when varying the action with respect to the gauge field, i.e. $\mathbf A\to \mathbf A + d\chi_{\mathbf A}$ and $\mathbf B\to \mathbf B+d\chi_{\mathbf B}$, we pick up a boundary term:
$$
\delta S_\text{CS}=\cdots+\frac{1}{2\pi}\int_{\R^2\times\{-1,1\}}\Tr(d\chi_{\mathbf A}\wedge d\mathbf B+d\chi_{\mathbf B}\wedge d\mathbf A). 
$$
Therefore, in order to have a consistent theory in the presence of boundaries, we must impose boundary conditions on the gauge field such that this term vanishes (see e.g. \cite{CWYI}). We accommodate for this by requiring that $\mathbf A=0$ on the upper boundary $\R^2\times \{1\}$ and $\mathbf B=0$ on the lower boundary $\R^2\times \{-1\}$. 

\subsection{The propagator}\label{sec:propagator}
As explained above the gauge field can propagate between states $\mathbf A^a(x)$ and $\mathbf B_b(y)$ for some $x,y\in M$ and $a,b\in \{1,\dots, \dim \mathfrak a\}$. The corresponding probability distribution is a two form ${P^a}_b(x,y)$ known as the propagator. It satisfies the following defining relations:
%In the figure, the edge going between the lines represents a propagator which 
\begin{equation}\label{eq:anti-sym}
%\braket{\mathbf A^a(x)\,\mathbf B_b(y)}=-\braket{\mathbf B_b(y)\,\mathbf A^a(x)},
{P^a}_b(x,y)=-{P_b}^a(y,x)
\end{equation}
%Define $P^{ab}\in \Omega^2(\R^3\times \R^3\setminus \Delta)$ by $P^{ab}(x,y)=\braket{\alpha^a(x)\beta^b(y)}$.
\begin{equation}\label{eq:prop}
%d\braket{\mathbf A^a(x)\,\mathbf B_b(y)}={\delta^a}_b\delta^{(3)}(x,y).
d{P^a}_b(x,y)={\delta^a}_b\delta^{(3)}(x,y).
\end{equation}
%where $\mathscr C(\g)\in\g\otimes\g$ is the Casimir element given by
%$$
%\mathscr C(\g)=\sum_{a\in B(\mathfrak a)}(a\otimes a^*+a^*\otimes a)
%$$
where $d$ is the differential operator and $\delta^{(3)}(x,y)$ is the Dirac delta function. Furthermore, the boundary conditions on the gauge field translate to the following constraint on the propagator: 
\begin{equation}\label{eq:bc}
{P^a}_b(x,y)=0 \text{ when } x\in \R^2\times \{1\} \text{ or } y\in\R^2\times \{-1\}.
\end{equation}
%A propagator satisfying the constraint in equation \eqref{eq:anti-sym}-\eqref{eq:bc} is constructed as follows: 
%Take $P$ to be the pull back of a two-form $\omega\in\Omega^2(S^2)$ via the map $\phi:(\R^3\times\R^3)\setminus\Delta\to S^2$ given~by 
%\begin{equation*}%\label{phi}
%    \phi(x,y)=\frac{y-x}{|y-x|}\,.
%\end{equation*}
%The boundary constraints in equation \eqref{eq:bc} can be accommodated by requiring $\omega$ to only be supported locally around the north pole $x_{np}=(0,0,1)$. 
%Concretely, we define
%\begin{align}
%    \omega\coloneqq f\vol_{S^2}\in\Omega^2(S^2)\,.
%\end{align}
%where $\vol_{S^2}$ is the unit volume form on $S^2$ given in terms of the coordinates on $\R^3$ by: $\vol_{S^2}=x\,dy\,dz+y\,dz\,dx+z\, dx\, dy$ and $f:S^2\to\mathbb R$ is a smooth function supported in a small neighbourhood around $x_{np}$ and normalized so that $\omega$ integrates to one on $S^2$.
Let $\phi:(\R^3\times\R^3)\setminus\Delta\to S^2$ be the map
\begin{equation*}%\label{phi}
    \phi(x,y)=\frac{y-x}{|y-x|}\,.
\end{equation*}
and define $\omega\in\Omega^2(S^2)$ by
\begin{align*}
    \omega\coloneqq f\vol_{S^2}\in\Omega^2(S^2),
\end{align*}
where $\vol_{S^2}$ is the unit volume form on $S^2$ given in terms of the coordinates on $\R^3$ by: $\vol_{S^2}=x\,dy\,dz+y\,dz\,dx+z\, dx\, dy$ and $f:S^2\to\mathbb R$ is a smooth function supported in a small neighbourhood around $x_{np}=(0,0,1)$ and normalized so that $\omega$ integrates to one on $S^2$. 

\begin{proposition} If we define $P\in\Omega^2(\R^3\times \R^3\setminus \Delta)$ by 
\begin{equation}\label{eq:propagator}
P=\phi^*\omega,
\end{equation} 
then ${P^a}_b(x,y)\coloneqq P(x,y){\delta^a}_b$ satisfies the constraints in equation \eqref{eq:anti-sym}-\eqref{eq:bc}.
\end{proposition}
\begin{proof}
    Since $\omega$ is a top-dimension form on $S^2$ it holds that $dP(x,y)=0$ away from the diagonal $x=y$. To see that $dP$ it is in fact the Dirac delta function we use Stokes' theorem: Fix some $x\in \R^3$ and let $B_x$ be the unit ball centered at $x$
\begin{equation*}
    \int_{y\in B_x}dP(x,y)=\int_{y\in B}dP(0,y)=\int_{y\in S^2}P(0,y)=\int_{y\in S^2}\omega(y)=1.\qedhere
\end{equation*}
\end{proof}
%In the following we refer to the two-form $P$ as the propagator. 
\subsection{Wilson lines}\label{sec:W-lines}
%\subsection{Leading order interactions}
With our choice of boundary conditions the global gauge symmetry of the action is completely broken. As a consequence, the theory admits a set of gauge invariant operators known as Wilson lines (see \cite{CWYI} for more details). For the present purpose we will think of a Wilson line simply as a proper embedding in $L:\R\hookrightarrow \R^2\times I$ parallel to the boundary,
%labeled by some representation of $\g$, 
along with a rule that a gauge field $\mathbf A^a$ (resp. $\mathbf B_a$) couples to $L$ by inserting a basis element $\xi_a$ (resp. $\zeta^a$) at the corresponding point in $L$.
%The weight is computed using a set of Feynman rules that can be derived from the Chern-Simons action in equation \eqref{CS action}.
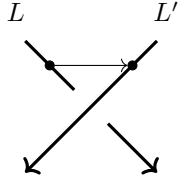
\begin{figure}[H]
    \centering
\begin{tikzpicture}[vertex/.style={draw,circle, fill=black, inner sep=1.2pt}]

\node (b1) at (-1,1)  {}; 
  \node (b2) at (-0.1,0.1)  {};
  \node (b3) at (0.1,-0.1)  {};
  \node (b4) at (1,-1)  {};
  \draw[very thick] (b1) node[above]{$L$} -- (b2);
  \draw[very thick,->] (b3) -- (b4);
  
  \node (c1) at (1,1)  {}; 
  \node (c2) at (-1,-1)  {};
  \draw[very thick,->] (c1)node[above]{$L'$} --  (c2);

  \node[vertex] (v1) at (-0.55,0.55) {};
   \node[vertex] (v2) at (0.55,0.55) {};
   \draw[->] (v1) -- node[above]{} (v2);
\end{tikzpicture} 
\caption{: The projection onto $\R^2$ of a pair of crossing Wilson lines. }
    \label{fig:crossing}
\end{figure}
%and the expectation value $\braket{LL'}$ is given in perturbation theory as a power expansion in the parameter $\hbar$. Each term in the expansion is represented by a set of weighted graph (Feynman diagrams). 
\noindent Consider for example a pair of Wilson lines $L$ and $L'$ supported at different points in $I$ and crossing in $\R^2$ as shown in figure \ref{fig:crossing}. The two Wilson lines interact by exchanging gauge bosons. The simplest (leading order) 
%diagram corresponds to 
interaction corresponds to a single gauge boson propagating between the lines. This interaction is illustrated in figure \ref{fig:crossing}, where the oriented edge represents a propagator. The corresponding amplitude is given~by
$$
\hbar \int_{x\in L, y\in L'}P(x,y){\delta^{a}}_b~\xi_a\otimes \zeta^b,,
$$
where $\hbar$ is a small expansion parameter. At higher orders in $\hbar$ we get interactions coming from the cubic interaction term in the Chern-Simons action in equation \eqref{eq:action}. Each interaction is represented by a directed graph (Feynman diagram) with three-valent interaction vertices in the bulk and one-valent vertices along the Wilson lines. The expectation value for the interaction is an element $\braket{LL'}\in \mathcal U(\g)^{\otimes 2}$ given as an perturbative expansion in $\hbar$ in terms of the set of Feynman diagrams:
\begin{equation*}
\braket{LL'}=\sum_{\Gamma}\hbar^{\ord(\Gamma)}\mathcal M(\Gamma),
\end{equation*}
where $\ord(\Gamma)$ is the number of edges of $\Gamma$ minus the number internal vertices and the weight (amplitude) $\mathcal M(\Gamma)$ is determined by the Feynman rules.
\begin{remark}\label{rmk:R}
    On the surface it appears that the expectation value $\braket{LL'}$ depends on the angle of crossing between the lines $L$ and $L'$. We will argue in section \ref{sec:conclusion} that $\braket{LL'}$ is in fact independent of the angle. For now we take this for given and define $\mathcal R\in\mathcal U(\g)^{\otimes 2}$ by
\begin{equation}\label{eq:R}
    \mathcal R \coloneqq \braket{LL'}.
\end{equation}
\end{remark}

\subsection{The $R$-matrix from crossing Wilson lines}
The goal of the remainder of this paper is to show that the element $\mathcal R$ is a quantum $R$-matrix, that is,
%at a crossing between a pair of Wilson lines 
it satisfies the Yang-Baxter equation \eqref{eq:YBE}. 
% that is
%$$
%R_{12}R_{13}R_{23}=R_{23}R_{13}R_{12}\,.
%$$
In this framework, the lines in the Yang-Baxter picure should be thought of as representing Wilson line operators supported at different points in $I$. With this as our motivation we define the following smooth family of proper embeddings:
\begin{definition}\label{def:W-lines} Let $L_t$ be a family of embeddings
$$
L_t:\coprod_{\alpha=1,2,3}\R_\alpha\hookrightarrow \R^2\times I,
$$
parametrized by $t\in[0,1]$, where $L_t|_{\R_\alpha}=L_{\alpha,t}:\R\hookrightarrow \R^2\times I$ is given by
\begin{equation*}
    \begin{aligned}
    L_{1,t}: s\mapsto (-s/\sqrt{2},s/\sqrt{2},-1/2)\ , \ \
    L_{2,t}: s\mapsto (t,s,0)\ , \ \
    L_{3,t}: s\mapsto (s/\sqrt{2},s/\sqrt{2},1/2).
    \end{aligned}
\end{equation*}
\end{definition}
The family of embeddings defined above is illustrated in figure \ref{fig:W-lines} which shows the projection onto $\R^2$. As $t$ increases, the lines $L_{1,t}$ and $L_{3,t}$ are held fixed while $L_{2,t}$ is dragged continuously over the crossing between the other two lines. 
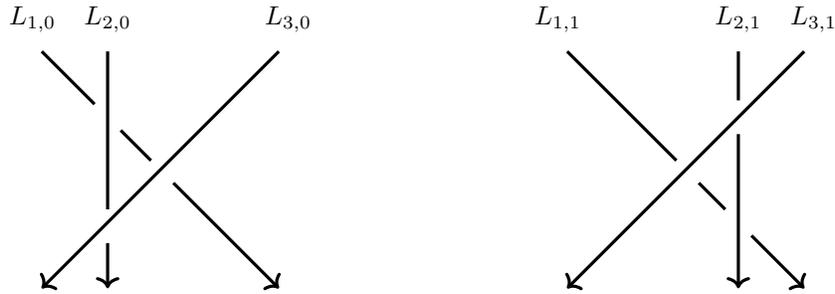
\begin{figure}[H]
    \centering
\begin{tikzpicture}[vertex/.style={draw,circle, fill=black, inner sep=1.5pt}]

  \node (a1) at (-0.7,1.7)  {}; 
  \node (a2) at (-0.7,-0.65)  {};
  \node (a3) at (-0.7,-0.85)  {}; 
  \node (a4) at (-0.7,-1.7)  {};
  \draw[very thick] (a1)node[above]{$L_{2,0}$} --  (a2);
  \draw[very thick,->] (a3) --  (a4);

  \node (b1) at (-1.7,1.7)  {}; 
  \node (b2) at (-0.75,0.75)  {};
   \node (b3) at (-0.65,0.65)  {}; 
  \node (b4) at (0,0)  {};
  \node (b5) at (0.05,-0.05)  {};
  \node (b6) at (1.7,-1.7)  {};
  \draw[very thick] (b1) node[above]{$L_{1,0}$} -- (b2);
  \draw[very thick] (b3) -- (b4);
  \draw[very thick,->] (b5) -- (b6);
  
  \node (c1) at (1.7,1.7)  {}; 
  \node (c2) at (-1.7,-1.7)  {};
  \draw[very thick,->] (c1)node[above]{$L_{3,0}$} --  (c2);

\end{tikzpicture} \hspace{2.5cm}
\begin{tikzpicture}[vertex/.style={draw,circle, fill=black, inner sep=1.5pt}]

  \node (a1) at (0.7,1.7)  {};
  \node (a2) at (0.7,0.8)  {};
  \node (a3) at (0.7,0.6)  {};
  \node (a4) at (0.7,-1.7)  {};
  \draw[very thick] (a1)node[above]{$L_{2,1}$} --  (a2);
  \draw[very thick,->] (a3) --  (a4);

  \node (b1) at (-1.7,1.7)  {}; 
  \node (b2) at (0,0)  {};
  \node (b3) at (0.05,-0.05)  {};
  \node (b4) at (0.65,-0.65)  {};
  \node (b5) at (0.75,-0.75)  {};
  \node (b6) at (1.7,-1.7)  {};
  \draw[very thick] (b1) node[above]{$L_{1,1}$} -- (b2);
  \draw[very thick] (b3) -- (b4);
  \draw[very thick,->] (b5) -- (b6);
  
  \node (c1) at (1.7,1.7)  {}; 
  \node (c2) at (-1.7,-1.7)  {};
  \draw[very thick,->] (c1)node[above]{$L_{3,1}$} --  (c2);

\end{tikzpicture} 
\caption{: The embedding space $L_t$ where $t\in [0,1]$.}
    \label{fig:W-lines}
\end{figure}
%We will compute this expectation value in perturbation theory as an expansion in terms of a set of weighted graphs called Feynman diagrams with vertices along the Wilson lines and in the ambient space: 
\noindent For each $t\in[0,1]$, the corresponding expectation value is an element $$\braket{L_t}\coloneqq\braket{L_{1,t}L_{2,t} L_{3,t}}\in \mathcal U(\g)^{\otimes 3}.
$$ 
The following section is dedicated to giving a precise definition of $\braket{L_t}$, which on the surface appears to depend on the parameter $t\in[0,1]$. The main objective of this paper is to show that $\braket{L_t}$ is in fact independent on $t$. Since the form of the propagator ensures that interactions only take place in a small neighbourhood around each crossing, this will imply the the expectation value of a pair of crossing Wilson lines is an $R$-matrix. A formal argument for this is given in section \ref{sec:admissible} below.
%is a power expansion in $\hbar$ given~by
%\begin{equation}\label{eq:graph sum}
%\braket{L_t}=\sum_{\Gamma\in\G}\hbar^{\ord(\Gamma)}\mathcal M_t(\Gamma),
%\end{equation}
%where $\G$ is the set of Feynman graphs, which will be defined in the following section. Notice that the amplitudes $\mathcal M_t(\Gamma)$ now depend on the parameter $t\in[0,1]$.

%As we shall see, the form of the propagator in equation \eqref{eq:prop} ensures that interactions can only take place in a small neighbourhood around each crossing. Thus, showing that the expectation value of a pair of crossing Wilson lines satisfies the Yang-Baxter equation amounts to showing that $\braket{L_1}-\braket{L_0}=0$. 
%Since we are interested in the how the expectation value of a set of Wilson lines changes under displacing the lines, we will in the remainder of the paper consider more generally a smooth family of $k$ proper embeddings $L_{1,t},\dots, L_{k,t}\subset \R^2\times I$ parametrized by $t\in[0,1]$. Then, for each $t\in[0,1]$, the expectation value $\braket{L_t}\coloneqq\braket{L_{1,t}\cdots L_{k,t}}\in \mathcal U(\g)^{\otimes k}$ is a power expansion in $\hbar$ given~by
%\begin{equation}\label{eq:sum}
%    \braket{L_t}=\sum_\Gamma \hbar^{\ord (\Gamma)} A_t(\Gamma)\in \mathcal U(\g)^{\otimes k},
%\end{equation}
%where the sum runs over all Feynman diagrams $\Gamma$ and $\ord(\Gamma)$ is defined as the number of edges minus the number of internal vertices of $\Gamma$. 

\section{Chern-Simons Perturbation Theory}\label{sec:Feynman}
%and present the Feynman rules for computing the amplitude $\mathcal M_t(\Gamma)$ associated to each graph. 
In this section we give a definition of the expectation value $\braket{L_t}$ in the formalism of perturbation theory. As mentioned, $\braket{L_t}$ is given by an expansion in $\hbar$ in terms of a set of weighted graphs called Feynman graphs which we define in subsection \ref{sec:graphs} below. %In particular, in section \ref{sec:graphs} we define the relevant set of graphs, and in section \ref{sec:conf} and \ref{F-rules} we define the associated weight (amplitude).

\subsection{Feynman graphs}\label{sec:graphs}
We here define the relevant set of graphs contributing to the expectation value $\braket{L_t}$. %These are graphs with three-valent interaction vertices in the bulk and one-valent vertices on Wilson lines. 
Given $m\in\mathbb Z_{\geq 0}$ and $\n=(n_1,n_2,n_3)$ a tuple of integers $n_\alpha\in\mathbb Z_{\geq 0}$, we first fix the data corresponding to the sets of $m$ internal (bulk) vertices and of $n_\alpha$ external vertices on the Wilson line $L_{\alpha,t}$, along with a set of half-edges incident on each vertex:%, with three half-edges incident on and internal vertex and a single half-edge incident on an external vertex:
%Given $m\in\mathbb Z_{\geq 0}$ and $\n=(n_1,n_2,n_3)$ a tuple of integers $n_\alpha\in\mathbb Z_{\geq 0}$, we define a set of graphs with fixed set of $m$ internal vertices and $n_\alpha$ external vertices on the Wilson line $L_{\alpha,t}$, along with a fixed set of half-edges at each vertex. We begin by defining this data:
%We first fix the data corresponding to the set of vertices and half-edges of the graphs. 

Let $n=\sum_\alpha n_\alpha$. We define a set $\mathcal V$ of vertices consisting of:
    \begin{enumerate}
        \item A set of internal vertices $V=\{v_1,\dots,v_m\}$.
        \item An set of external vertices $W_\alpha$ in each Wilson line $L_{\alpha,t}$, given by:
        $$
        W_1=\{w_1,\dots, w_{n_1}\} \ , \ \  W_2=\{w_{n_1+1},\dots, w_{n_1+n_2}\} \ , \ \ W_3=\{w_{n_1+n_2+1},\dots, w_{n}\}. 
        $$
        We write $W=\bigcup_{\alpha=1}^3W_\alpha$ and $\underline W=(W_1,W_2,W_3)$.
    \end{enumerate}
Moreover, we define a set $\mathcal H$ of half-edges consisting of:
    \begin{enumerate}
        \item A set of half-edges $\{h_i^1,h_i^2,h_i^3\}$ for each internal vertex $v_i\in V$.
        \item A single half-edge $h_j$ for each external vertex $w_j\in W$.
    \end{enumerate}
Finally, we denote by $s:\mathcal H\to \mathcal V$ the source map $s(h_i^k)= v_i$ and $s(h_j)= w_j$. 

With the above data fixed, the only data needed to define a graph is an involution of the set of half-edges to form edges. In addition, we want the definition of a Feynman graph to include an orientation of the edges and a Lie algebra labeling of the half-edges. This leads to the following definition:

\begin{definition}[Feynman graphs]\label{def:graphs} A Feynman graph $\Gamma\in \G_{m,\n}$ is defined by the following data:
\begin{enumerate}[(i)]
    \item A free involution $\iota:\mathcal H\to \mathcal H$ such that, if $\iota(h^k_{i})=h^l_j$ then $i\neq j$. A pair $\{h,\iota(h)\}$ is called an edge and we denote the set of edges by $E(\Gamma)$.
    \item An orientation of the edges corresponding to an ordering $(h,h')$ of each pair $\{h,h'\}\in E(\Gamma)$.
    \item An assignment $\tau:\mathcal H\to \mathcal B(\mathfrak a)\cup \mathcal B(\mathfrak a^*)$ such that if $(h,h')\in E(\Gamma)$ then $\tau(h)\in \mathcal B(\mathfrak a)$ and $\tau(h')=\tau(h)^*\in \mathcal B(\mathfrak a^*)$.
\end{enumerate}
\end{definition}
We write $\G=\bigcup_{m,\n}\G_{m,\n}$ for the collection of all Feynman graphs. When writing the expectation $\braket{L_t}$ we only wish to sum over isomorphism classes of Feynman graphs. Let us therefore make precise what it means for two Feynman graphs to be isomorphic. 
\begin{definition}\label{iso}  Two graphs $\Gamma,\Gamma'\in\G_{m,\n}$ are said to be isomorphic, and we write $\Gamma\sim\Gamma'$, if there are bijections
    $$F_{\mathcal V}:\mathcal V\to\mathcal V \ , \ \ F_{\mathcal H}:{\mathcal H}\to {\mathcal H}$$
such that:
\begin{enumerate}[(i)]
    \item $F_{\mathcal V}$ acts as the identity map on the set of external vertices.
    %$F_{\mathcal V}\big|_{W}:W\to W$ is the identity map.
    %order preserving. 
    \item $(F_{\mathcal V},F_{\mathcal H})$ is a graph isomorphism: $F_{\mathcal V}\circ s=s\circ F_{\mathcal H}$ and $F_{\mathcal H}\circ \iota=\iota'\circ F_{\mathcal H}$.
    \item $(F_{\mathcal V},F_{\mathcal H})$ preserves the edge orientation: If $(h,h')\in {E}(\Gamma)$ then $(F_{\mathcal H}(h),F_{\mathcal H}(h'))\in E(\Gamma')$. %$$\rho'_{\{F_H(h),F_H(h')\}}\circ \left. F_H\right|_e=\rho_{\{h,h'\}}\,.$$
    \item $(F_{\mathcal V},F_{\mathcal H})$ preserves the Lie algebra decoration of edges: $\tau(h,h')=\tau'(F_{\mathcal H}(h),F_{\mathcal H}(h'))$.
\end{enumerate}
%Furthermore, to ease notation we suppress the indices $\n$ and $m$ when the context does not allow for confusion. 
\end{definition}
\subsection{The configuration space of vertices}\label{sec:conf}
%In the remainder of the paper we consider more generally a smooth family of proper embeddings of $k$ copies of $\R$ into $\R^2\times I$ parametrized by $t\in[0,1]$:
%\begin{equation}\label{eq:L_t}
%L_t:\coprod_{i=1}^k\R_i\hookrightarrow \R^2\times I
%\end{equation}
%and we write $L_t\big|_{\R_i}\coloneqq L_{i,t}$. Notice that for each $t\in[0,1]$ the expectation value is now an element $\braket{L_t}=\braket{L_{1,t}\cdots L_{k,t}}\in \mathcal U(\g)^{\otimes k}$. %We begin by defining the relevant set of Feynman graphs below.
%given by a power expansion in terms of a set of weighted graphs. 
%In general
%\begin{equation}\label{eq:expansion}
%\braket{L_t}=\sum_{\Gamma} \hbar^{\ord\Gamma} A_t(\Gamma)\in \mathcal U(\g)^{\otimes k}
%\end{equation}
We wish to consider the space of embeddings of the vertices $V\cup W$ of Feynman graphs into $\R^2\times I$, such that for each $\alpha\in \{1,2,3\}$ the set of external vertices $W_\alpha$ maps to the Wilson line $L_{\alpha,t}$. We here give a formal definition of the space in question, following the definition given by Bott and Taubes in \cite{BT}. 

Let %a general manifold $M$ and 
$S$ be some ordered set. We denote by $\Conf_S(\R^2\times I)$ the configuration space of $|S|$ ordered points in $\R^2\times I$, i.e. the space of injections $S\hookrightarrow \R^2\times I$. Moreover, we denote by $\Conf_S(\R)$ the space of injections $S\hookrightarrow \R$ such that the points in $S$ are placed in increasing order along $\R$. Recall definition \ref{def:W-lines} and observe that an embedding $L_{t,\alpha}:\R\hookrightarrow \R^2\times I$ induces an embedding of configuration spaces
$
\Conf_{W_\alpha}(\R)\hookrightarrow \Conf_{W_\alpha}(\R^2\times I).
$
Hence we have a map:
\begin{equation}\label{def:A} 
\mathscr L: \prod_{\alpha=1}^k\Conf_{W_\alpha}(\R)\times [0,1] \longrightarrow  \Conf_{W}(\R^2\times I)\,.
\end{equation} 
%Furthermore, let $\Conf_{S}(\R)$ be the set of injection $S\hookrightarrow \R$ such that the points in $S$ are arranged in increasing order along $\R$. 
The relevant configuration space $\Conf_{V,\underline W}$ is now defined as the pullback:
\begin{equation}\label{Conf_nm}
    {\begin{tikzcd}
    \Conf_{V,\underline W}\ar[d] \ar[r] & \Conf_{V\cup W}(\R^2\times I) \ar[d,"\pi"] \\
    \prod_{\alpha=1}^3 \Conf_{W_\alpha}(\R)\times [0,1] \ar[r,"\mathscr L"] & \Conf_W(\R^2\times I)~.
    \end{tikzcd}}
\end{equation}
In particular, we can describe $\Conf_{V,\underline{W}}$ as the set of points $(t,q,p)$, where $t\in [0,1]$, $q\in \prod_{\alpha=1}^3 \Conf_{W_\alpha}(\R)$ and $p\in \Conf_{V}\big(\R^2\times I\setminus\big\{\mathscr L(q,t)(w_i)\big\}_{w_i\in W}\big)$. 

Notice that we have a projection 
$$
\Conf_{V,\underline W}\to [0,1]
$$
via the map on the left-hand side of the diagram \eqref{Conf_nm}. We write $\Conf^t_{V,\underline W}$ for the fiber of this map over $t\in[0,1]$.
%Moreover, 
\subsection{The expectation value} \label{F-rules}
%and $t\in[0,1]$ 
%an element $\lambda(\Gamma)\in\mathcal U(\g)^{\otimes 3}$. 
%\subsection{The Feynman rules} \label{F-rules}
We are now equipped to present the Feynman rules that determines the amplitude $\mathcal{M}_t(\Gamma)$ associated to any $\Gamma\in \G$ and $t\in[0,1]$. Our first step is to define a differential form of $\lambda(\Gamma)$ on $\Conf_{V,\underline W}$ as follows: For each edge $e=(h,h')\in E(\Gamma)$, let
$\phi_e: \Conf_{V\cup W}(\R^2\times I)\to S^2$ be the map
\begin{equation*}
    \phi_{e}(x)=\frac{x(s(h'))-x(s(h))}{|x(s(h'))-x(s(h))|},
\end{equation*} 
where $s:\mathcal H\to \mathcal V$ is the source map (see section \ref{sec:graphs}). Furthermore, let
$
\Phi_{e}: \Conf_{V,\underline W}\to S^2
$ 
be the pull back of $\phi_e$ to $\Conf_{V,\underline W}$ along the map in the top row of diagram \eqref{Conf_nm} and write $P_e=\Phi_e^*\omega\in \Omega^2(\Conf_{V,\underline W})$. We define %$\lambda(\Gamma)\in \Omega^{d-1}(\Conf_{V,\underline W})$ by:
\begin{equation}
\lambda(\Gamma)\coloneqq \bigwedge_{e\in E(\Gamma)}P_e.
\end{equation}
Notice that the degree of $\lambda(\Gamma)$ is $2|E|=3|V|+|W|$ and hence $\lambda(\Gamma)$ is a form of co-dimension one on $\Conf_{V,\underline W}$.
%where $d=3|V|+|W|$ is the dimension of $\Conf_{V,\underline W}$. 
Moreover, we associate to $\Gamma$ a Lie-algebra factor $c(\Gamma)\in \mathcal U(\g)^{\otimes 3}$ as follows:
\begin{enumerate} [(i)]
\item For each internal vertex $v_i$ we multiply by a factor: 
%\begin{figure}[H]
\begin{center}
\begin{tikzpicture}[vertex/.style={draw,circle, fill=black, inner sep=1.5pt}]
  \node[vertex] (v0) at (0,0)  {}; 
  \node (v1) at (0,-0.9)  {};
  \node (v2) at (-0.8,0.6) {};
  \node (v3) at (0.8,0.6) {};
  \node (a) at (1.5,-0.2) {$\sim$};
  \node (b) at (3.7,-0.2) {$\big<[\tau(h_i^1),\tau(h_i^2)],\tau(h_i^3)\big>$};
  \draw (v0) --  (v1)node[right]{};  
  \draw (v0) --  (v2)node[above]{};
  \draw (v0) --  (v3)node[above]{};

  \node at (0.25,-0.15) {$v_i$};
  \end{tikzpicture} 
\end{center}

\item For each Wilson line $L_\alpha$ we get an element of $\mathcal U(\g)$ given by:
\begin{center}
\begin{tikzpicture}[vertex/.style={draw,circle, fill=black, inner sep=1.5pt}]
\node (a) at (-1.6,0)  {}; 
  \node (b) at (1.7,0)  {};
  \node at (2.5,0.2) {$\sim$};
  \node at (5.3,0.2) {$\cdots \tau(h_j)\tau(h_{j+1})\tau(h_{j+2}) \cdots$};
  \draw[ultra thick,->] (a) --  (b);
  
  \node[vertex] (w) at (0,0) {};
  \node (v) at (0,1) {};
  \draw (w)node[below]{$w_{j+1}$} -- node[right]{}(v);

  \node[vertex] (w1) at (-0.9,0) {};
  \node (v1) at (-0.9,1) {};
  \draw (w1)node[below]{$w_j$} -- node[right]{}(v1);

  \node[vertex] (w2) at (0.9,0) {};
  \node (v2) at (0.9,1) {};
  \draw (w2)node[below]{$w_{j+2}$} -- node[right]{}(v2);

  \node at (-1.3,0.5) {$\cdots$};
  \node at (1.4,0.5) {$\cdots$};
  
\end{tikzpicture} 
\end{center}
\end{enumerate}
In other words, $c(\Gamma)$ is given by
\begin{equation}\label{eq:color}
c(\Gamma)=\prod_{i=1}^{m}\big<[\tau(h_i^1),\tau(h_i^2)],\tau(h_i^3)\big>~\prod_{j=1}^{n_1}\tau(h_j)\otimes \prod_{k=n_1+1}^{n_2}\tau(h_k)\otimes \prod_{l=n_1+n_2}^{n}\tau(h_l).
\end{equation}
Given $t\in [0,1]$ and $\Gamma\in \G$ we now wish to define the amplitude $\mathcal{M}_t(\Gamma)$ as the integral of the element $\lambda(\Gamma)c(\Gamma)$ over the configuration space of vertices $\Conf^t_{V,\underline W}$.
%\begin{equation}\label{eq:F-int}
%    \mathcal M_t(\Gamma)=\int_{\Conf^t_{V,\underline W}}\lambda(\Gamma)~c(\Gamma).
%\end{equation}
%\subsection{The Expectation Value} 
However, to properly define such an integral 
%the integral in equation \eqref{eq:F-int} 
we must equip the configuration space with a suitable orientation form. Specifically, the orientation form in question must ensure that integrals are invariant under isomorphisms of $\Gamma\in \G$. Furthermore, the anti-symmetry relation in equation \eqref{eq:anti-sym} implies that changing the orientation of an edge must reverse the sign of orientation of the configurations space.
%Following along the lines of Altschuler and Freidel \cite{AF} we can associate an orientation form $\Or(\Gamma)$ to $\Conf_{V,\underline W}$ satisfying the these constraints as follows: 

Given a point $(q,p)\in \Conf^t_{V,\underline W}$ we write $p_i= p(v_i)\in \R^2\times I$ and $q_j= q(w_j)\in \R$. Then, a small neighbourhood of $(p,q)\in\Conf^t_{V,\underline W}$ has local coordinates $t\in\R$, $(p_{i}^{1},p_{i}^{2},p_{i}^{3})\in\R^3$ for each internal vertex $v_i\in V$ and $q_{j}\in\R$ for each external vertex $w_j\in W$.  %We can associate an orientation form $\Or(\Gamma)$ to $\Conf_{V,\underline W}$ satisfying the the above constraints as follows: 
\begin{definition}\label{def:Or} 
    Let $g:\mathcal H\to \R$ be the map $g(h_i^k)=p_i^{k}$ and $g(h_j)=q_j$. For each $\Gamma\in \G$ we define an orientation form on $\Conf^t_{V, \underline W}$ by
    \begin{equation*}  \Or(\Gamma)=\bigwedge_{(h,h')\in E(\Gamma)} \big(dg(h)\wedge dg(h')\big).
    \end{equation*}
    %\begin{equation*}  \Or(\Gamma)=\Bigg(\bigwedge_{(h^i_\alpha,h^j_\beta)\in E(\Gamma)} dp_{\beta}^{j}\wedge dp_{\alpha}^{i}\Bigg)\wedge\Bigg(\bigwedge_{(h_\gamma, h_\delta^k)\in E(\Gamma)}dp_\delta^{k}\wedge dq_\gamma\Bigg)\wedge\Bigg( \bigwedge_{(h_\kappa, h_\rho)\in E(\Gamma)} dq_\rho\wedge dq_\kappa\Bigg)\wedge dt\,.
    %\end{equation*}
\end{definition}
In the following we use the notation $\Conf^t(\Gamma)$ to denote the configuration space $\Conf^t_{V,\underline W}$ equipped with the orientation form $\Or(\Gamma)$. Similarly we denote by $\Conf(\Gamma)$ the configuration space $\Conf_{V,\underline W}$ equipped with the orientation form $\Or(\Gamma)\wedge dt$. We now define 
\begin{equation}\label{eq:F-int}
\mathcal M_t(\Gamma)=\int_{\Conf^t(\Gamma)}\lambda(\Gamma)\,c(\Gamma).
\end{equation}

\begin{proposition}\label{Or invariance} The Feynman amplitude $\mathcal M_t(\Gamma)$ in equation \eqref{eq:F-int} is invariant under isomorphisms of $\Gamma$.
\end{proposition}
\begin{proof}
    By definition \ref{iso}, any isomorphism of $\Gamma$ is given by relabeling the internal vertices and permuting the set of half-edges at each internal vertex. Since the definition of $\mathcal M_t(\Gamma)$ does not depend on the labeling of vertices, we consider an isomorphism that permutes the half-edges $\{h_i^1,h_i^2,h_i^3\}$ incident to some $v_i\in V$. If the permutation is odd then the sign of $\Or(\Gamma)$ is reversed. On the other hand, since the structure constants are totally anti-symmetric, also $c(\Gamma)$ reverses its sign, thus leaving the overall sign of $\mathcal M_t(\Gamma)$ unchanged. 
    %Similarly, if the permutation even, neither $\Or(\Gamma)$ nor $c(\Gamma)$ changes sign. 
\end{proof}
We are now finally ready to give a precise definition of the expectation value $\braket{L_t}$
\begin{definition}
We define
\begin{equation}\label{eq:graph sum} \braket{L_t}=\sum_{\Gamma\in\G/\sim}\hbar^{\ord(\Gamma)}\mathcal M_t(\Gamma) \in \mathcal U(\g)^{\otimes 3},
\end{equation}
where the sum runs over isomorphism classes of Feynman graphs, and $\ord(\Gamma)$ is the number of edges minus the number of internal vertices of $\Gamma$.
\end{definition}

\subsection{Admissible Feynman graphs}\label{sec:admissible}
Only a limited set of Feynman graphs has a non-vanishing contribution to the sum in equation~\eqref{eq:graph sum}. In fact, recall that $\g=\mathfrak a\oplus \mathfrak a^*$ is defined by the following non-trivial brackets: 
$$
[\xi_a,\xi_b]={f^c}_{ab}\xi_c \ , \ \ [\zeta^a,\xi_b]={f^a}_{cb}\zeta^c.
$$
With this definition, the coefficient $\braket{[\tau(h_i^1),\tau(h_i^2)],\tau(h_i^3)}$ associated to an internal vertex $v_i$ is only non-zero when $v_i$ has exactly one incoming and two outgoing edges, as shown in figure \ref{fig:vertex}.
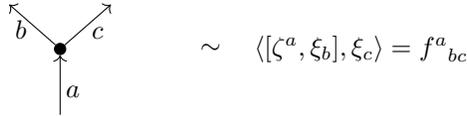
\begin{figure}[H]
\centering
\begin{tikzpicture}[vertex/.style={draw,circle, fill=black, inner sep=1.5pt}]
  \node[vertex] (v0) at (0,0)  {}; 
  \node (v1) at (0,-1)  {};
  \node (v2) at (-0.8,0.7) {};
  \node (v3) at (0.8,0.7) {};
  \node (a) at (2,0) {$\sim$};
  \node (b) at (4,0) {$\braket{[\zeta^a,\xi_b],\xi_c}={f^a}_{bc}$};
  \draw[<-] (v0) -- (v1);  
  \draw[->] (v0) -- (v2);
  \draw[->] (v0) --  (v3);
  \node at (-0.52,0.26) {$b$};
  \node at (0.5,0.22) {$c$};
  \node at (0.17,-0.57) {$a$};

  %\draw[->] (0,0.4) arc (90:330:0.4);
\end{tikzpicture}
\caption{: The only allowed internal vertex}
\label{fig:vertex}
\end{figure}
\noindent Moreover, we get no contributions from graphs that have an oriented cycles as shown in figure \ref{fig:loops} (a) or from graphs that have an oriented path that ends and begins on the same Wilson line as shown in figure \ref{fig:loops} (b). This follows from the definition of the propagator $P_e=\phi^*_e\omega$. In fact, because $\omega$ is only non-zero in a small neighbourhood of the north pole, $\lambda(\Gamma)$ is only supported in a neighbourhood of $\Conf_{V,\underline W}$ where all edges in $\R^2\times I$ point strictly upwards along $I$. Hence $\lambda(\Gamma)$ vanishes everywhere for the graphs in figure \ref{fig:loops}. 
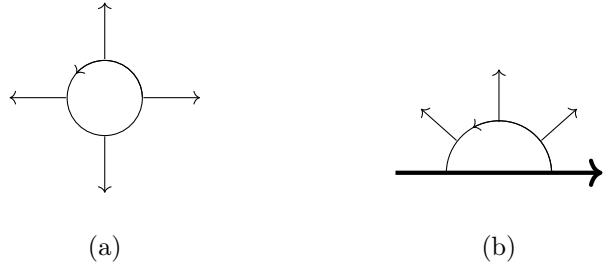
\begin{figure}[H]
    \centering
\begin{tikzpicture}[vertex/.style={draw,circle, fill=black, inner sep=1.2pt}]
  \draw[->] (0.5,0) arc (0.5:140:0.5);
  \draw[] (0.5,0) arc (0.5:360:0.5);
  
  \node (v1) at (0,0.39) {};
  \node (w1) at (0,1.39) {};
  \draw[->] (v1) -- (w1);
  
  \node (v2) at (-0.39,0) {};
  \node (w2) at (-1.39,0) {};
  \draw[->] (v2) -- (w2);

  \node (v3) at (0,-0.39) {};
  \node (w3) at (0,-1.39) {};
  \draw[->] (v3) -- (w3);
  
  \node (v4) at (0.39,0) {};
  \node (w4) at (1.39,0) {};
  \draw[->] (v4) -- (w4); 

  \node at (0,-2) {(a)};

\end{tikzpicture}\hspace{2cm}
\begin{tikzpicture}
  \node (a) at (-1.5,0)  {}; 
  \node (b) at (1.5,0)  {};
  \draw[ultra thick,->] (a) --  (b);
  \draw[->] (0.7,0) arc (0.7:120:0.7);
  \draw[] (0.7,0) arc (0.7:179:0.7);
  
  \node (v1) at (0.44,0.32) {};
  \node (w1) at (1.16,0.96) {};
  \draw[->] (v1) -- (w1);

  \node (v2) at (-0.44,0.32) {};
  \node (w2) at (-1.16,0.96) {};
  \draw[->] (v2) -- (w2);

  \node (v3) at (0,0.55) {};
  \node (w3) at (0,1.5) {};
  \draw[->] (v3) -- (w3);

  \node at (0,-1) {(b)};
  
\end{tikzpicture} 
    \caption{:~Non-contributing Feynman graphs}
    \label{fig:loops}
\end{figure}
\noindent 
The above discussion can be summarized to give the following proposition:
\begin{proposition}\label{prop:forest}
The only Feynman diagrams contributing to the sum in equation \eqref{eq:graph sum} are forests with edges in $\R^2\times I$ pointing strictly upwards along $I$ and roots and leafs connected to the Wilson lines (see figure \ref{fig:forest}). 
\end{proposition}
\begin{figure}[H]
    \centering
\begin{tikzpicture}[vertex/.style={draw,circle, fill=black, inner sep=1.3pt}]
  \node[vertex] (v0) at (-1,-2)  {}; 
  \node[vertex] (v1) at (-1,-1.3) {};
  \node[vertex] (v2) at (-1.5,0) {};
  \node[vertex] (v3) at (-0.5,-0.7) {};
  \node[vertex] (v4) at (-1,2) {};
  \node[vertex] (v5) at (-0.2,0) {};
  %\node (a) at (2,0) {$\sim$};
  %\node (b) at (4,0) {$\braket{[t_a,t_b],t^*_c}={f_{abc}}$};
  \draw[->] (v0) --  (v1);
  \draw[->] (v1) --  (v2);
  \draw[->] (v1) --  (v3);
  \draw[->] (v3) --  (v4);
  \draw[->] (v3) --  (v5);

  \node[vertex] (w0) at (1,0)  {}; 
  \node[vertex] (w1) at (1,1) {};
  \node[vertex] (w2) at (0.5,2) {};
  \node[vertex] (w3) at (1.5,2) {};
  \draw[->] (w0) --  (w1);
  \draw[->] (w1) --  (w2);
  \draw[->] (w1) --  (w3);

  \node (a1) at (-2.5,2)  {};
  \node (a2) at (2.5,2)  {};
  \draw[very thick,->] (a1) -- (a2)node[right]{$L_{3,t}$};

  \node (b1) at (-2.5,0)  {};
  \node (b2) at (2.5,0)  {};
  \draw[very thick,->] (b1) -- (b2)node[right]{$L_{2,t}$};

  \node (c1) at (-2.5,-2)  {};
  \node (c2) at (2.5,-2)  {};
  \draw[very thick,->] (c1) -- (c2)node[right]{$L_{1,t}$};
\end{tikzpicture}
 \caption{: Example of an admissible Feynman graph}
    \label{fig:forest}
\end{figure}
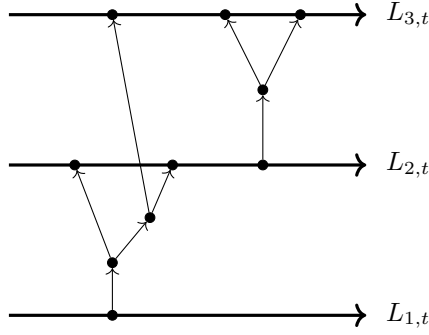
\noindent It follows from proposition \ref{prop:forest} that a given \textit{connected} Feynman graph $\Gamma$ connects at least two Wilson lines. Again using the fact that $\omega$ is only non-zero in a small neighbourhood of the north pole, it follows that the associated differential form $\lambda(\Gamma)$ only has support in a small neighbourhood of $\R^2$ around the crossing between the corresponding Wilson lines. Recall from remark \ref{rmk:R} of section \ref{sec:W-lines} that we denoted the (angle independent) expectation value of a pair of crossing Wilson lines by $\mathcal R\in \mathcal U(\g)^{\otimes 2}$.  For each $i,j\in\{1,2,3\}$ with $i\neq j$ let 
$
\rho_{ij}:\mathcal U(\g)^{\otimes 2}\to \mathcal U(\g)^{\otimes 3}
$ 
be the map defined in section \ref{sec:YBE}, i.e.
\begin{align*}
    \rho_{12}(a\otimes b)=a\otimes b\otimes 1 \, , \ \rho_{13}(a\otimes b)=a\otimes 1\otimes b \, ,  \ \rho_{23}(a\otimes b)=1\otimes a\otimes b.
\end{align*}
and write $\mathcal R_{ij}=\rho_{ij}(\mathcal R)\in \mathcal U(\g)^{\otimes 3}$. By the above discussion we now have the following lemma:
\begin{lemma}\label{lm:infinity} 
   $\braket{L_0}$ and $\braket{L_1}$ takes the form  
   $$
   \braket{L_0}=\mathcal R_{12}\mathcal R_{13}\mathcal R_{23} \ \text{ and } \ \braket{L_1}=\mathcal R_{23}\mathcal R_{13}\mathcal R_{12}\,.
   $$
\end{lemma}
\noindent The situation is illustrated in figure \ref{fig:R}. The dotted circle indicates the area where the interaction matrix $\mathcal R_{ij}$ acts.
%\begin{proof}
 %    Let $\Gamma\in \G$ and let $\wt\Gamma$ be a connected subgraph of $\Gamma$ with vertices along $L_{\alpha,t}$ and $L_{\beta,t}$ for $\alpha,\beta\in\{1,2,3\}$. By the above discussion it holds that $\lambda(\Gamma)$ is only non-zero when the vertices of $\wt \Gamma$ are close to the point where $L_{\alpha,t}$ and $L_{\beta,t}$ cross in $\R^2$. 
%\end{proof}
\begin{figure}[H]
    \centering
\begin{tikzpicture}[vertex/.style={draw,circle, fill=black, inner sep=1.5pt}]

  \node (a1) at (-0.7,1.7)  {}; 
  \node (a2) at (-0.7,-0.67)  {};
  \node (a3) at (-0.7,-0.78)  {}; 
  \node (a4) at (-0.7,-1.7)  {};
  \draw[very thick] (a1)node[above]{$L_{2,0}$} --  (a2);
  \draw[very thick,->] (a3) --  (a4);

  \node (b1) at (-1.7,1.7)  {}; 
  \node (b2) at (-0.72,0.72)  {};
   \node (b3) at (-0.68,0.68)  {}; 
  \node (b4) at (0,0)  {};
  \node (b5) at (0.02,-0.02)  {};
  \node (b6) at (1.7,-1.7)  {};
  \draw[very thick] (b1) node[above]{$L_{1,0}$} -- (b2);
  \draw[very thick] (b3) -- (b4);
  \draw[very thick,->] (b5) -- (b6);
  
  \node (c1) at (1.7,1.7)  {}; 
  \node (c2) at (-1.7,-1.7)  {};
  \draw[very thick,->] (c1)node[above]{$L_{3,0}$} --  (c2);

  \draw[dotted,thick] (0.31,0.0) arc (0.3:360:0.3);
  \draw[dotted,thick] (-0.38,0.7) arc (0.3:360:0.3);
  \draw[dotted,thick] (-0.38,-0.7) arc (0.3:360:0.3);

  \node at (-1.4,0.7) {$\mathcal R_{12}$};
   \node at (0.8,0) {$\mathcal R_{13}$};
  \node at (-1.4,-0.7) {$\mathcal R_{23}$};

\end{tikzpicture} \hspace{2.5cm}
\begin{tikzpicture}[vertex/.style={draw,circle, fill=black, inner sep=1.5pt}]

  \node (a1) at (0.7,1.7)  {};
  \node (a2) at (0.7,0.75)  {};
  \node (a3) at (0.7,0.63)  {};
  \node (a4) at (0.7,-1.7)  {};
  \draw[very thick] (a1)node[above]{$L_{2,1}$} --  (a2);
  \draw[very thick,->] (a3) --  (a4);

  \node (b1) at (-1.7,1.7)  {}; 
  \node (b2) at (0,0)  {};
  \node (b3) at (0.02,-0.02)  {};
  \node (b4) at (0.68,-0.68)  {};
  \node (b5) at (0.72,-0.72)  {};
  \node (b6) at (1.7,-1.7)  {};
  \draw[very thick] (b1) node[above]{$L_{1,1}$} -- (b2);
  \draw[very thick] (b3) -- (b4);
  \draw[very thick,->] (b5) -- (b6);
  
  \node (c1) at (1.7,1.7)  {}; 
  \node (c2) at (-1.7,-1.7)  {};
  \draw[very thick,->] (c1)node[above]{$L_{3,1}$} --  (c2);

  \draw[dotted,thick] (0.31,0.0) arc (0.3:360:0.3);
  \draw[dotted,thick] (1.01,0.7) arc (0.3:360:0.3);
  \draw[dotted,thick] (1.01,-0.7) arc (0.3:360:0.3);

  \node at (1.5,0.7) {$\mathcal R_{12}$};
   \node at (-0.7,0) {$\mathcal R_{13}$};
  \node at (1.5,-0.7) {$\mathcal R_{23}$};
\end{tikzpicture} 
\caption{:~A graphical illustration of lemma \ref{lm:infinity}.}
    \label{fig:R}
\end{figure}
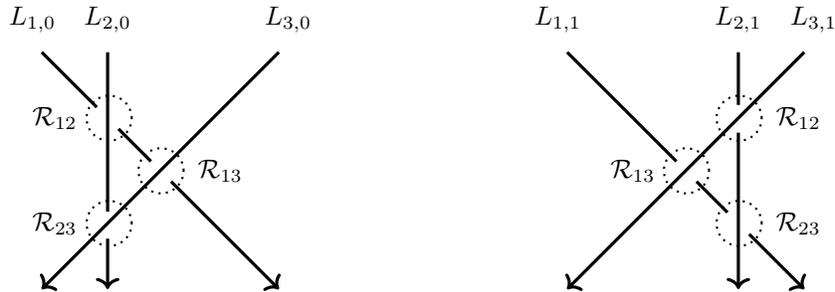
%\begin{proof}
%    Recall the definition of the propagator $P_e=\phi_e^*\omega$ where $\omega\in\Omega^2(S^2)$ is only non-zero in a small neighbourhood of the north pole. It follows that $\lambda(\Gamma)$ is support in a neighbourhood of $\Conf(\Gamma)$ where all edges of $\Gamma$ ``point upwards'' along $I$. The lemma now follows from property that In fact, for $\lambda(\Gamma)$ to be non-zero in $\Conf(\Gamma)$, any internal vertex must belong to an oriented path in $\Gamma$ connecting two different Wilson lines and for the edges to point upwards along $I$ all vertices along such a path must lie in a small neighbourhood of where the lines cross in~$\R^2$. 
%\end{proof}

\section{Finiteness of the Integrals}\label{sec:finiteness}
Because the propagator $P_e=\phi_e^*\omega$ is only defined away from the diagonal
%, the differential forms $\lambda(\Gamma)$ do not have compact support on $\Conf_{V,\underline W}$ and 
it is not immediately clear that the Feynman integrals in equation \eqref{eq:F-int} converges in the limit when vertices come together. In fact, the finiteness of Feynman integrals in Chern-Simons theory on a general three-manifold was proven in \cite{AS} by Axelrod and Singer, using a configuration space compactification closely related to the Fulton-MacPherson compactification \cite{FM}, and in \cite{BT} this was extended by Bott and Taubes to Chern-Simons theory in the presence of Wilson lines. 
%In this section we recall the relevant configuration space compactification of $\Conf_{V,\underline W}$. 
For the present purpose these results can be assembled to give the following theorem: 
\begin{theorem}
    There is a partial compactification $\overline \Conf_{V,\underline W}$ of the configuration spaces $\Conf_{V,\underline W}$ for subsets of vertices coming together, such that the compactified space is a manifold with corners and the differential forms $\lambda(\Gamma)$ are smooth forms with compact support on $\Conf_{V,\underline W}$.
\end{theorem}
%For the purpose of the present paper it will only be necessary to describe the  for vertices coming together. The corresponding co-dimension one boundary
In this compactification boundary strata are defined using spherical blow-ups along diagonals where subsets of vertices come together. In subsections \ref{sec:int_strata} and \ref{sec:ext_strata} below we give a full description of the corresponding boundary strata of co-dimension one, each coming from a single subset of vertices all coming together at the same speed. Denoting by $\partial \Conf_{V,\underline W}$ the corresponding co-dimension one boundary it holds that $\partial\Conf_{V,\underline W}$ is given by the disjoint union of the following strata:
\begin{itemize}
    \item For each $S\subset V$ we get boundary stratum $\partial_S\Conf_{V,\underline W}$ corresponding to the vertices $S$ coming together.
    %in $\Conf_{\underline A,B}$.
    \item For each $\alpha\in\{1,2,3\}$, $S\subset V$ and $T\subset W_\alpha$ with $T\neq \emptyset$ we get a boundary stratum $\partial_{S,T}\Conf_{V,\underline W}$ corresponding to vertices $S\cup T$ coming together on the line $L_{\alpha,t}$.
    %in $\Conf_{\mathbf n,m}$. 
\end{itemize}
The reader is referred to \cite{AS}, \cite{Sinha} and the appendix of \cite{BT} for details on the strata of higher co-dimension, which correspond to collapsing nested subsets of vertices.
%\begin{equation*}
%\begin{aligned}
%\partial \overline \Conf_{\underline A,B} = \bigg(\bigcup_{S\subset B} \partial_{S}\Conf_{\underline A,B} \bigg)&\cup \bigg(\bigcup_{i=1}^3\bigcup_{\substack{T\subset A \\S\subset B}}\partial_{T,S}\Conf_{\underline A,B}\bigg)\,.
%\end{aligned}
%\end{equation*}
\subsection{Boundary strata for internal collisions}\label{sec:int_strata}
We begin by describing the boundary strata corresponding to a subset $S\subset V$ of internal vertices coming together. Recall that given a point $(t,q,p)\in \Conf_{V,\underline W}$ we use the notation $p_i=p(v_i)$ and $q_j=q(w_j)$. Let $i_0\coloneqq \min\{i\}_{v_i\in S}$ and write $v_0\coloneqq v_{i_0}$ and $p_0\coloneqq p_{i_0}$. Furthermore, let $d_{\min}$ be the minimal distance between $p_{0}$ and a vertex in $\{p_i\}_{v_i\in V\setminus S}\cup \{q_j\}_{w_j\in W}$. We can define a neighbourhood $U\subset \Conf_{V,\underline W}$ where the vertices in $S$ are close together and far from all other vertices as follows:
%where the points $\{p_\alpha\}_{v_\alpha\in S}$ are close to each other and far from all other points as follows:
$$
U=\bigg\{(t,q,p)\in \Conf_{V,\underline W}~\bigg|~\bigg(\sum_{v_i\in S}|p_{0}-p_i|^2\bigg)^{1/2}<\eta d_{\min}\bigg\}\,,
$$
where $\eta >0$ is small. Given any point $(t,q,p)\in U$ we can now write
\begin{equation}\label{eq:coord change1}
p_i=p_{0}+rd_{\min}u_\alpha \ , \ \ v_i\in S\setminus \{v_{0}\},
\end{equation}
where $u_\alpha\in \R^3$ and $r \in (0,\eta)$ are uniquely determined by the conditions:
\begin{itemize}
    \item $\sum_i|u_i|^2=1$,
    %\item $\sum_\alpha u_\alpha=0$
    \item $u_i\neq u_j$ for $i\neq j$.
\end{itemize}
\begin{definition} Let $G<\Homeo(\R^3)$ to be group of scalings and translations in $\R^3$. We define $$C_S\coloneqq\Conf_S(\R^3)/G$$ where $G$ acts on $\Conf_S(\R^3)$ by translating and/or scaling all points simultaneously. 
\end{definition}
The points $(u_i)_{i\in S\setminus \{v_{0}\}}$ then determines a set of coordinates on the space $C_S$ and hence the change of coordinates in equation \eqref{eq:coord change1} determines a diffeomorphism
\begin{equation}\label{eq:U}
U\cong C_S\times\Conf_{(V\setminus S)\cup\{v_0\},\underline W}\times (0,\eta).
\end{equation}
%given by $\varphi(t,q,p)=(x,(u_\alpha),r)$, where 
%$$
%x=\big(t,p_0,\{p_{\alpha}\}_{v_\alpha\in B\setminus S},\{q_\beta\}_{w_\beta \in A}\big) \ , \ \ r=\frac{1}{d_{\min}}\Big(\sum_{v_\alpha\in S}|p_0-p_\alpha|^2\Big)^{1/2}\ , \ \ u_\alpha=\frac{p_\alpha-p_0}{rd_{\min}}.
%$$
The boundary stratum corresponding to the vertices $\{p_i\}_{i\in S}$ coming together is obtained by including the $r=0$ in the interval on the right-hand side of equation \eqref{eq:U}. Hence
\begin{equation}\label{eq:internal}
    \partial_S\Conf_{V,\underline W}=C_S\times\Conf_{(V\setminus S)\cup\{v_0\},\underline W}\,.
\end{equation}
\subsection{Boundary strata for external collisions}\label{sec:ext_strata}
We now describe the boundary strata corresponding to a subset of both internal and external vertices coming together on one of the Wilson lines.
%$L_{\alpha,t}$ for some $\alpha\in \{1,2,3\}$. 
%The construction is essentially analogous to that of internal collisions given in the previous section with a few modifications. 
%and consider a neighbourhood $V\subset \Conf_{\underline A,B}$ where the vertices $\{p_\alpha, q_\beta\}_{v_\alpha\in S, w_\beta\in T}$ are close to each other and far from all other vertices. 
Let $S\subset V$ and $T\subset W_\alpha$ for some $\alpha\in \{1,2,3\}$ and let $\mathbf e_\alpha$ be the unit vector pointing along $L_{\alpha,t}$ (notice that $\mathbf e_\alpha$ does not depend on $t$). Given a point $(t,q,p)\in \Conf_{V,\underline W}$ we use the following notation:
\begin{itemize}
    \item $\braket{p_i,\mathbf e_\alpha}$ is the projection of $p_i$ onto $L_{\alpha,t}$,
    \item $j_0=\min\{j\}_{w_j\in T}$ and we write $w_0\coloneqq w_{j_0}$ and $q_0\coloneqq q_{j_0}$
    %is the center of mass of $\{\braket{p_\alpha,\mathbf e_i}\}_{v_\alpha\in S}\cup \{q_\beta\}_{w_\beta\in T}$,
    \item $d_{\min}$ is the minimal distance between $L_{\alpha,t}(q_0)$ and a vertex in $(V\setminus S)\cup (W\setminus T)$.
    %$\{p_\gamma\}_{v_\gamma\in B\setminus S}\cup\{L_{i,t}(q_\delta)\}_{w_\delta\in A\setminus T}$.
\end{itemize}
We can define a neighbourhood $V\subset \Conf_{V,\underline W}$ where the vertices in $S\cup T$ are close together and far from all other vertices as follows:
%where the points $\{p_\alpha\}_{v_\alpha\in S}\cup \{L_{i,t}(q_\beta)\}_{w_\beta\in T}$ are close to each other and far from all other points as follows:
$$
V=\bigg\{(t,q,p)\in \Conf_{V,\underline W}~\bigg|~\bigg(\sum_{v_i\in T}|q_0-\braket{p_i,\mathbf e_i}|^2+\sum_{w_j\in T}|q_0-q_j|^2\bigg)^{1/2}<\eta d_{\min}\bigg\}\,.
$$
Given any $(t,q,p)\in V$, $v_i\in S$ and $w_j\in T$ we can write:
\begin{equation}\label{coord_change2}
\begin{aligned}
p_i&=L_{\alpha,t}(q_0)+rd_{\min}u_i \ , \ \ v_i\in S\\
q_j&=q_0+rd_{\min}a_j %\implies L_{i,t}(q_\beta)=L_{i,t}(q_0)+rd_{\min}a_\beta\mathbf e_i
\ , \ \ w_\beta\in T\setminus \{w_0\}
\end{aligned}
\end{equation}
for unique $r\in (0,\eta)$, $u_i\in \R^3$ and $a_j\in \R$ subject to the conditions:
\begin{itemize}
    %\item $\sum_\alpha\braket{u_\alpha,\mathbf e_i}+\sum_\beta a_\beta=0$,
    \item $\sum_i |\braket{u_i,\mathbf e_\alpha}|^2+\sum_j|a_j|^2=1$,
    \item $u_i\neq u_j$, $a_i\neq a_j$ and $u_i\neq a_j \mathbf e_\alpha$ when $i\neq j$.
\end{itemize}
\begin{definition}
Let $\Conf_{S,T}(L,\R^3)$ be the configuration space with points in the bulk and along the line $L\subset \R^3$.
%$L_\alpha=\{s\mathbf e_\alpha~|~s\in \R\}\subset \R^3$. 
Concretely, $\Conf_{S,T}(L,\R^3)$ is defined as the pullback: 
    \begin{equation}
        {\begin{tikzcd}
        \Conf_{S,T}(L,\R^3)\ar[d] \ar[r] &  \Conf_{S\cup T}(\R^3) \ar[d] \\
        \Conf_{T}(\R) \ar[r,hook,"L"] & \Conf_T(\R^3)~.
    \end{tikzcd}}
\end{equation}
%where the embedding $\Conf_T(\R)\hookrightarrow \Conf_T(\R^3)$ is the induced by the embedding $L\subset \R^3$.
%\label{def:Conf_TS}
Moreover, let $G'<\Homeo(\R^3)$ be the subgroup of scalings and translations along $L$. We define 
$$
C_{S,T}\coloneqq \Conf_{S,T}(L,\R^3)/G'
$$ 
where $G'$ acts on $\Conf_{T,S}$ by translating and/or scaling all points simultaneously. 
\end{definition}
The points $\big\{(u_i)_{v_i\in S},(a_j)_{w_j\in T\setminus \{w_0\}}\big\}$ determine a set of coordinates on the space $C_{S,T}$ defined above and hence the change of coordinates in equation \eqref{coord_change2} determines a diffeomorphism
\begin{equation}\label{eq:V}
V\cong \Conf_{V\setminus S,\underline{W}'}\times C_{S,T}\times  (0,\eta),
\end{equation}
where $\underline W'$ is obtained from $\underline W$ by substituting $W_\alpha$ with $(W_\alpha\setminus T)\cup\{w_0\}$.
The boundary stratum corresponding to the vertices $S\cup T$ coming together is obtained by including the $r=0$ in the interval on the right-hand side of equation \eqref{eq:V}. Hence
\begin{equation}\label{eq:bdry2}
    \partial_{S,T}\Conf_{V,\underline W}\cong\Conf_{V\setminus S,\underline{W}'}\times C_{S,T}.
\end{equation} 
\section{Stokes' Theorem}
The remainder of this paper is dedicated to proving theorem \ref{thm:1}, namely that $\Delta_t \braket{L_t}=\braket{L_1}-{\braket{L_0}=0}.$ 
To this aim we will use the below proposition.
\begin{proposition} \label{prop:Stokes} Let $\partial\Conf(\Gamma)$ be the co-dimension one boundary in the Axelrod-Singer compactification. Then
\begin{equation}\label{eq:boundary}
\Delta_t \braket{L_t} = \sum_{\Gamma}{\hbar^{\ord(\Gamma)}}\int_{\partial \Conf(\Gamma)} \lambda(\Gamma)\,c(\Gamma)\,.
\end{equation}
\end{proposition}
\begin{proof}
    Observe that the total co-dimension one boundary of $\overline \Conf(\Gamma)$ is given by the union of boundary components coming from
\begin{enumerate}
    \item The boundary $\partial \Conf(\Gamma)$ corresponding to subsets of vertices coming together.
    \item The boundaries $\Conf^1(\Gamma)$ and $\Conf^0(\Gamma)$ corresponding to $t=0$ and $t=1$.
    \item The boundaries coming from an internal vertex reaching $\R^2\times \{-1\}$ or $\R^2\times \{1\}$.
\end{enumerate}
By proposition \ref{prop:forest} and lemma \ref{lm:infinity} in section \ref{sec:admissible} it holds for any $\Gamma\in \G$ that $\lambda(\Gamma)$ has compact support in $\overline{\Conf}(\Gamma)$ and vanishes on the boundary corresponding to case 3 in the above. Moreover, since the propagator is a closed form on the interior of $\overline \Conf(\Gamma)$ it holds that $d\lambda(\Gamma)=0$. The following version of Stokes' theorem now applies:
\begin{equation}\label{eq:stokes}
0=\int_{\overline \Conf(\Gamma)}d\lambda(\Gamma)=\int_{\partial \Conf(\Gamma)} \lambda(\Gamma) +\int_{\Conf^0(\Gamma)}\lambda(\Gamma)- \int_{\Conf^1(\Gamma)}\lambda(\Gamma)\,.
\end{equation}
Inserting equation \eqref{eq:stokes} into the expression for $\braket{L_t}$ in equation \eqref{eq:graph sum} the proposition follows.
\end{proof}
%$$
%\Delta_t \mathcal M_t(\Gamma)=\bigg(\int_{\Conf^1(\Gamma)}\lambda(\Gamma)-\int_{\Conf^0(\Gamma)}\lambda(\Gamma)\bigg)c(\Gamma)=\int_{\partial \Conf(\Gamma)} \lambda(\Gamma)\,c(\Gamma),
%$$ and hence 
Proving theorem \ref{thm:1} therefore amounts to showing that the sum of all boundary integrals in equation \eqref{eq:boundary} vanishes. By the construction in the previous section we have
\begin{align*}
\int_{\partial \Conf(\Gamma)} \lambda(\Gamma)\,c(\Gamma) = \sum_{S} \int_{\partial_S \Conf(\Gamma)} \lambda(\Gamma)\,c(\Gamma) +\sum_{S,T}\int_{\partial_{S,T} \Conf(\Gamma)} \lambda(\Gamma)\,c(\Gamma)\,.
\end{align*}
%Where the sums run over all subsets $S\subset  V$ and $T\subset W_\alpha$, $\alpha\in\{1,2,3\}$. 
%We show in section \ref{sec:int} and \ref{sec:ext} that the contribution to the sum in equation \eqref{eq:boundary} coming from boundary strata corresponding to respectively internal and external collisions vanishes. 
\section{Vanishing Theorems}\label{sec:vanishing}
This section contains the proof of theorem \ref{thm:1} via a series of vanishing results for the boundary integrals in equation \eqref{eq:boundary}. These results are variations of the vanishing theorems of Bott and Taubes \cite{BT}. Concretely, in section \ref{sec:int} we prove the vanishing of boundary integrals coming from internal collisions and in section \ref{sec:ext} we prove the vanishing of boundary integrals coming from external collisions (collisions along a Wilson line) .
\subsection{Vanishing theorems for internal collisions}\label{sec:int}
%The goal of this section is to prove the following theorem:
\begin{theorem}\label{thm:int}
   The boundary integrals contributing to equation \eqref{eq:boundary} coming from internal collisions vanishes, that is
\begin{equation}\label{eq:internal_sum}
    \sum_{\Gamma}{\hbar^{\ord(\Gamma)}}\sum_{S} \int_{\partial_S \Conf(\Gamma)} \lambda(\Gamma)\,c(\Gamma) =0\,.
\end{equation}
\end{theorem}
\paragraph{Notation:} Given $\Gamma\in \G$ and $S\subset V$, we denote by $\Gamma_S$ the sub-graph of $\Gamma$ spanned by the vertices in $S$ and by $\delta_S\Gamma$ the graph obtained from $\Gamma$ by collapsing $\Gamma_S$ to a single internal vertex $v_0$. Then
$$
\partial_S\Conf(\Gamma)=C_S\times \Conf(\delta_S\Gamma).
$$

% We write $\Conf(\Gamma_0)$ for the configuration space of vertices in $\Gamma_0$. 
Observe that $\lambda(\Gamma)$ splits into a product 
$
\lambda(\Gamma)=\lambda_1\wedge \lambda_2
$
where $\lambda_1$ is constructed from edges in $\Gamma_S$ and $\lambda_2$ is constructed from the remaining edges. In order to prove theorem \ref{thm:int} we will need the following lemma. 
\begin{lemma}\label{prop:splitting}
Upon restricting to $\partial_S\Conf(\Gamma)$, the form $\lambda_1$ factors through the projection
$$
\pi_1:\partial_{S}\Conf(\Gamma)\to C_{S}
$$
and the form $\lambda_2$ factors through the projection 
$$
\pi_2:\partial_{S}\Conf(\Gamma)\to \Conf(\delta_S\Gamma).
$$
%If $e\in E(\Gamma)$ connects two vertices in $S$ then the restriction of $P_e=\Phi_e^*\omega$ to $\partial_{S}\Conf_{\mathbf n,m}$ factors through the projection 
%Otherwise, the restriction of $P_e$ factors through the projection 
\end{lemma}
\begin{proof}
%$\partial_{T,S}\Conf_{\underline A,B}$ with $T\subset A_2$ (the other case is similar) and 
The proposition follows from the change of coordinates in equation \eqref{eq:coord change1}. In fact, if $e$ connects two internal vertices $v_i,v_j\in S$ we have
\begin{equation*}\label{eq:extensionI}
    \Phi_e(x)=\frac{p_j-p_i}{|p_j-p_i|}=\frac{u_j-u_i}{|u_j-u_i|}\,,
\end{equation*}
%As before, this expression is independent of $r$ and hence $\Phi_e$ extends continuously to the boundary stratum. 
which implies that $\Phi_e$ and thereby $P_e$ factors through the projection $\pi_1$. On the other hand, if $e$ connects a vertex $v_i\in V\setminus S$ and a vertex $v_j\in S$ we have
\begin{equation*}  \Phi_e(x)=\frac{p_j-p_i}{|p_j-p_i|}=\frac{p_0+ru_j-p_i}{\big|p_0+ru_j-p_i\big|}\to \frac{p_0-p_i}{|p_0-p_i|} \ \text{when} \ r\to 0,
\end{equation*}
and hence $P_e$ factors through the projection $\pi_2$. 
\end{proof}
%By lemma \ref{prop:splitting} we can write
%$$
%\int_{\partial_{S}\Conf(\Gamma)}\lambda(\Gamma)=\int_{C_{S}}\lambda_{S}(\Gamma)\int_{\Conf(\Gamma_0)}\lambda(\Gamma_0)\,,
%$$
We write 
$$\lambda_1|_{\partial_{S}\Conf(\Gamma)}=\pi_1^*\,\wt\lambda(\Gamma_{S}) \ \text{ and } \ \lambda_2|_{\partial_{S}\Conf(\Gamma)}=\pi_2^*\,\lambda(\delta_S\Gamma).
$$
%Let $S\subset \{1,\dots, m\}$ label the subset of vertices coming together and
\begin{corollary}\label{lm:dim2}
 Given $\Gamma\in \G$ and $S\subset V$, let $\eta_S(\Gamma)$ be the number of edges connecting a vertex in $S$ with a vertex in $(V\cup W)\setminus S$. The contribution to equation \eqref{eq:internal_sum} from the boundary stratum $\partial_S\Conf(\Gamma)$ vanishes unless $\eta_S(\Gamma)=4$.    
\end{corollary}
\begin{proof}
%Recall the form of the corresponding boundary stratum: 
%$$
%\partial_{S\cup T}\Conf_{n,m}=C_{s,t}\times \Conf_{n-t+1,m-s}\,.
%$$
%We write $\pi_1$ for the projection onto the first factor and $\pi_2$ for the projection onto the second factor. Given $\Gamma\in \G_{n,m}$ 
By counting the number of edges connecting vertices in $S$ one finds 
$
\deg \wt\lambda(\Gamma_{S}) = 3|S|-\eta_S(\Gamma)
$. On the other hand, $\dim C_{S}= 3|S|-4$, and hence $\wt\lambda(\Gamma_{S})$ vanishes unless $\eta_S(\Gamma)\geq 4$. By a similar argument $\lambda(\Gamma_{S,T})$ vanishes on the boundary stratum unless $\eta_S(\Gamma)\leq 4$. 
\end{proof}
\begin{lemma}\label{thm:hidden_int}
    The contribution to equation \eqref{eq:boundary} coming from boundary strata where more than two internal vertices come together vanishes. 
\end{lemma}
\begin{proof}
    This follows directly from corollary \ref{lm:dim2} and proposition \ref{prop:forest}, since collapsing more than two internal vertices in a forest creates a vertex of valence greater than four. 
\end{proof}
%We are now ready to prove the vanishing theorems for internal collisions. We first consider the case of two internal vertices coming together. The following vanishing theorem is known as the IHX-relation.
%These boundary strata are known in the paper of Bott and Taubes as \textit{principal faces}.
%The vanishing in this case follows from relation known as the IHX-relations: {\color{red}(...)}
%\subsection{Two internal vertices coming together}
The following lemma is known as the IHX relations. 
\begin{lemma}\label{thm:IHX} The contribution to equation \eqref{eq:internal_sum} coming from boundary strata where two internal vertices come together vanishes.
\end{lemma}
\begin{proof} 
Let $\Gamma_0$ be a graph which has a single four-valent internal vertex $v_0$ with one incoming and three outgoing edges, and with all other vertices three- and one-valent. There are exactly three graphs $\Gamma_{1},\Gamma_{2},\Gamma_{3}\in\G$ that identify with $\Gamma_0$ when collapsing two internal vertices. These graphs are shown in figure \ref{fig:IHX}, where we imagine that all vertices and edges outside the encircled area are held fixed:
\begin{figure}[H]
    \centering
\begin{tikzpicture}[vertex/.style={draw,circle, fill=black, inner sep=1.5pt}]
\node (w0) at (0,0.3)  {$p_j$}; 
  \node (w1) at (0,-1.3)  {$p_i$};
  \node[vertex] (v0) at (0,0)  {}; 
  \node[vertex] (v1) at (0,-1)  {};
  \node (v2) at (-1,0.7) {};
  \node (v3) at (1,0.7) {};
  \node (v4) at (1,-1.7) {};
  \node (v5) at (-1,-1.7) {};
  \node (v6) at (0,-2.5) {${\Gamma_1}$};
  \draw[->] (v1) --  (v0);  
  \draw[->] (v0) --  (v2)node[above]{$a$};
  \draw[->] (v0) --  (v3)node[above]{$b$};
  \draw[->] (v1) --  (v4)node[below]{$d$};
  \draw[->] (v5)node[below]{$e$} --  (v1);
  \draw[dotted,thick] (1.1,-0.5) arc (1.1:360:1.1);
\end{tikzpicture} \hspace{1cm} \begin{tikzpicture}[vertex/.style={draw,circle, fill=black, inner sep=1.5pt}]
\node (w0) at (0.8,-0.5)  {$p_j$}; 
  \node (w1) at (-0.8,-0.5)  {$p_i$};
  \node[vertex] (v0) at (0.5,-0.5)  {}; 
  \node[vertex] (v1) at (-0.5,-0.5)  {};
  \node (v2) at (-1,0.7) {};
  \node (v3) at (1,0.7) {};
  \node (v4) at (1,-1.7) {};
  \node (v5) at (-1,-1.7) {};
  \node (v6) at (0,-2.5) {${\Gamma_2}$};
  \draw[->] (v1) --  (v0);  
  \draw[->] (v1) --  (v2)node[above]{$a$};
  \draw[->] (v0) --  (v3)node[above]{$b$};
  \draw[->] (v0) --  (v4)node[below]{$d$};
  \draw[->] (v5)node[below]{$e$} --  (v1);
  \draw[dotted,thick] (1.1,-0.5) arc (1.1:360:1.1);
\end{tikzpicture} \hspace{1cm} \begin{tikzpicture}[vertex/.style={draw,circle, fill=black, inner sep=1.5pt}]
\node (w0) at (0.8,-0.5)  {$p_j$}; 
  \node (w1) at (-0.8,-0.5)  {$p_i$};
  \node[vertex] (v0) at (0.5,-0.5)  {}; 
  \node[vertex] (v1) at (-0.5,-0.5)  {};
  \node (v6) at (0,-2.5) {${\Gamma_3}$};
  \node (v2) at (-1,0.7) {};
  \node (v3) at (1,0.7) {};
  \node (v4) at (1,-1.7) {};
  \node (v5) at (-1,-1.7) {};
  \draw[->] (v1) --  (v0);  
  \draw[->] (v0) --  (v2)node[above]{$a$};
  \draw[->] (v1) --  (v3)node[above]{$b$};
  \draw[->] (v0) --  (v4)node[below]{$d$};
  \draw[->] (v5)node[below]{$e$} --  (v1);
  \draw[dotted,thick] (1.1,-0.5) arc (1.1:360:1.1);
\end{tikzpicture} \hspace{0.5cm}
\begin{tikzpicture}
  \node (a1) at (0,0)  {};
  \node (a2) at (0,2)  {$\longrightarrow$};
\end{tikzpicture} \hspace{0.5cm}
\begin{tikzpicture}[vertex/.style={draw,circle, fill=black, inner sep=1.5pt}]
\node (w0) at (0.3,-0.5)  {$\,p_0$}; 
  %\node (w1) at (-0.8,-0.5)  {$v_\alpha$};
  \node[vertex] (v0) at (0,-0.5)  {}; 
  %\node[vertex] (v1) at (-0.5,-0.5)  {};
  \node (v2) at (-1,0.7) {};
  \node (v3) at (1,0.7) {};
  \node (v4) at (1,-1.7) {};
  \node (v5) at (-1,-1.7) {};
  \node (v6) at (0,-2.5) {${\Gamma_0}$};
  %\draw[->] (v1) --  (v0);  
  \draw[->] (v0) --  (v2)node[above]{$a$};
  \draw[->] (v0) --  (v3)node[above]{$b$};
  \draw[->] (v0) --  (v4)node[below]{$d$};
  \draw[->] (v5)node[below]{$e$} --  (v0);
  \draw[dotted,thick] (1.1,-0.5) arc (1.1:360:1.1);
\end{tikzpicture}
%\hspace{0.3cm}
%\begin{tikzpicture}
%  \node (a1) at (0,0)  {};
%  \node (a2) at (0,2)  {$\longrightarrow$};
%\end{tikzpicture} \hspace{0.3cm}
%\begin{tikzpicture}[vertex/.style={draw,circle, fill=black, inner sep=1.5pt}]
%\node (w0) at (0.3,-0.5)  {$v_0$}; 
%  \node[vertex] (v0) at (0,-0.5)  {}; 
%  \node (v2) at (-1,0.7) {};
%  \node (v3) at (1,0.7) {};
%  \node (v4) at (1,-1.7) {};
 % \node (v5) at (-1,-1.7) {};
  %\node (v6) at (0,-2.5) {${\Gamma_0}$};
  %\draw[->] (v0) --  (v2)node[above]{$a$};
  %\draw[->] (v0) --  (v3)node[above]{$b$};
  %\draw[->] (v0) --  (v4)node[below]{$d$};
  %\draw[->] (v5)node[below]{$e$} --  (v0);
  %\draw[dotted] (1.1,-0.5) arc (1.1:360:1.1);
%\end{tikzpicture}
\caption{:~The IHX-relation}
\label{fig:IHX}
\end{figure}
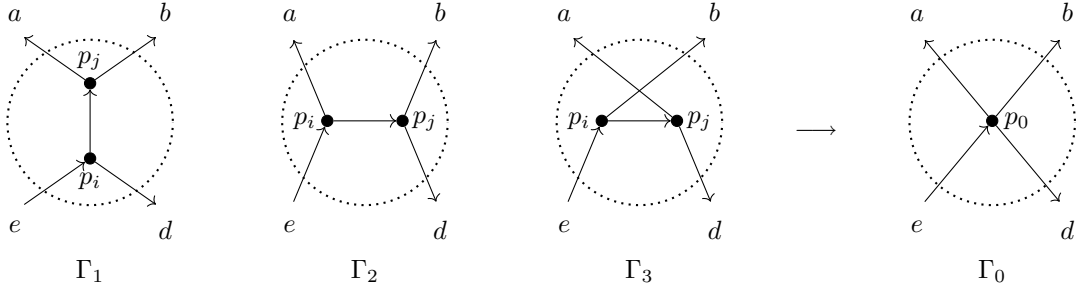
%Let $\lambda(\Gamma_0)\in \Omega^*(\Conf_{\mathbf{n},m-1})$ be given by $$\lambda(\Gamma_0)=\bigwedge_{e\in E(\Gamma_0)}P_e.$$
%Given $\Gamma\in \G_{\n,m}$, let $\Gamma_0$ be the graph obtained from collapsing the vertices $v_\alpha$ and $v_\beta$ into a single vertex $v_0$ and erasing all loops. In the diagrams below all edges and vertices outside of the dotted circle are fixed. Then there are three possible graphs $\Gamma_1$, $\Gamma_2$ and $\Gamma_3$. 
%Collapsing $p_i$ and $p_j$ in each $\Gamma_k$, $k\in\{1,2,3\}$ into a single vertex $p_0$, we obtain a graph $\Gamma_0$ which has a single four-valent internal vertex. 
\noindent The boundary stratum corresponding to collapsing $p_i$ and $p_j$ is given by: 
$$
\partial_{\{v_i,v_j\}}\Conf(\Gamma_k)=C_{\{v_i,v_j\}}\times \Conf(\Gamma_0)\cong S^2\times \Conf(\Gamma_0),
$$
%By lemma \ref{prop:splitting2} the three $\lambda(\Gamma_k)$'s restrict to the same form on the boundary:  $\pi^*_1\wt\lambda_1\wedge \pi_2^*\wt \lambda_2$
%namely
%$$
%\pi^*_1\omega\wedge \pi^*_2(\lambda(\Gamma_0)),
%$$
%where $\wt\lambda_1=\omega$ and $\wt\lambda_2=\lambda(\Gamma_0)$ is given by
%\begin{align*}
%\lambda(\Gamma_0)=\bigwedge_{e\in E(\Gamma_0)}P_e.
%\end{align*}
for $k=1,2,3$. If we choose the ordering of half edges in each graph to be clockwise, it follows from definition~\ref{def:Or} that 
$$
\Or(\Gamma_1)=-\Or(\Gamma_2)=\Or(\Gamma_3).
$$ 
Hence, the contribution to the sum in equation \eqref{eq:internal_sum} coming from this boundary stratum takes the form:
\begin{align}
\int_{S^2}\omega\int_{\Conf(\Gamma_0)}\lambda(\Gamma_0)\,c(\Gamma_0),
\end{align}
where $c(\Gamma_0)$ obtained from applying the usual Feynman rules to all three- and one-valent vertices of $\Gamma_0$, and assigning to the four-valent vertex $p_0$ the factor:
$$
({f^c}_{ab}{f^e}_{cd}-{f^c}_{bd}{f^e}_{ac}+{f^c}_{ad}{f^e}_{bc}),
$$
which vanishes by Jacobi identity for the structure constants. This proves the theorem.% On the other hand, for any $\Gamma_0$ as in figure \ref{fig:IHX}, there are exactly three graphs $\Gamma_{1,2,3}\in \G$ that maps to $\Gamma_0$ upon collapsing two internal vertices the theorem follows. 
%This relation is usually in the literature referred to as the IHX-relation. 
\end{proof}
By combining lemma \ref{thm:hidden_int} and \ref{thm:IHX} we have now proved theorem \ref{thm:int}. In section \ref{sec:ext} below we show the similar vanishing theorems for external collisions. Many of the arguments are repetitions of those given above.
\subsection{Vanishing theorems for external collisions}\label{sec:ext}
\begin{theorem}\label{thm:ext}
    The boundary integrals contributing to equation \eqref{eq:boundary} coming from external collisions vanishes, that is 
\begin{equation}\label{eq:external_sum}
    \sum_{\Gamma}{\hbar^{\ord(\Gamma)}}\sum_{S,T} \int_{\partial_{S,T} \Conf(\Gamma)} \lambda(\Gamma)\,c(\Gamma) =0.
\end{equation}
\end{theorem}
\paragraph{Notation:} Given $\Gamma\in \G$, 
%with vertex set $\underline A\cup B$ 
$S\subset V$ and $T\subset W_\alpha$ for some $\alpha\in\{1,2,3\}$, denote by $\Gamma_{S, T}$ the subgraph of $\Gamma$ spanned by the vertices in $S\cup T$ and by $\delta_{S,T}\Gamma$ the graph obtained from $\Gamma$ by collapsing $\Gamma_{S, T}$ to a single external vertex $w_0$. Then
$$
\partial_{S,T}\Conf(\Gamma)=C_{S,T}\times \Conf(\Gamma_{S,T}).
$$

We begin by proving the equivalents of lemma \ref{prop:splitting} and corollary \ref{lm:dim2} in the case of external collisions. As in section \ref{sec:int} we can write $\lambda(\Gamma)=\lambda_1\wedge\lambda_2$ where $\lambda_1$ is constructed from the edges in $\Gamma_{S, T}$ and $\lambda_2$ is constructed from the remaining edges. 

\begin{lemma}\label{prop:splitting2} 
Upon restricting to $\partial_{S,T}\Conf(\Gamma)$, the form $\lambda_1$ factors through the projection
$$
\wt\pi_1:\partial_{S,T}\Conf(\Gamma)\to C_{S,T}
$$
and the form $\lambda_2$ factors through the projection 
$$
\wt \pi_2:\partial_{S,T}\Conf(\Gamma)\to \Conf(\delta_{S,T}\Gamma).
$$  
%Let $\Gamma\in \G$. If $e\in E(\Gamma)$ connects two vertices in $S\cup T$ then the restriction of $P_e=\Phi_e^*\omega$ to $\partial_{T,S}\Conf(\Gamma)$ factors through the projection 
%$$
%\wt\pi_1:\partial_{S,T}\Conf(\Gamma)\to C_{S,T}\,.
%$$ 
%Otherwise, the restriction of $P_e$ factors through the projection 
%$$
%\wt \pi_2:\partial_{S,T}\Conf(\Gamma)\to \Conf(\Gamma_0)
%$$
\end{lemma}
\begin{proof}
%$\partial_{T,S}\Conf_{\underline A,B}$ with $T\subset A_2$ (the other case is similar) and 
%This expression does not depend in the parameter $r\in(0,1)$ and thus $\Phi_e$ extends continuously to the boundary stratum which corresponds to the limit $r\to 0$. 
Let $e$ be an edge connecting a vertices $v_i\in S$ and $w_j\in T$. Then with the coordinate change in equation \eqref{coord_change2} we have
%Observe first that
%\begin{equation*}
%L_{2,t}(q_0)=L_{2,0}(q_0)+t\hat x
%\end{equation*}
%where $\hat x$ is the unit vector along the $x$-axis {\color{red}(orthogonal to $\hat a_{2}$ parallel to the boundary)}. Hence, 
%With the coordinate change in equation \eqref{coord_change2} 
\begin{equation*}\label{eq:extensionII}
\Phi_e(x)=\frac{L_{\alpha,t}(q_j)-p_i}{\big|L_{\alpha,t}(q_j)-p_i\big|}=\frac{a_j\mathbf e_\alpha-u_i}{|a_j\mathbf e_\alpha-u_i|}\,,
\end{equation*}
%As before, this expression is independent of $r$ and hence $\Phi_e$ extends continuously to the boundary stratum. 
which implies that $\Phi_e$ and thereby $P_e$ factors through the projection $\wt\pi_1$. On the other hand, if $e$ connects a vertex $v_i\in V\setminus S$ and a vertex $v_j\in S$ we have
\begin{equation*}
    \Phi_e(x)=\frac{p_j-p_i}{|p_j-p_i|}=\frac{L_{\alpha,t}(q_0)+ru_j-p_i}{\big|L_{\alpha,t}(q_0)+ru_j-p_i\big|}\to \frac{L_{\alpha,t}(q_0)-p_i}{|L_{\alpha,t}(q_0)-p_i|} \ \text{when} \ r\to 0
\end{equation*}
and hence $\Phi_e$ factors through the projection $\wt\pi_2$. The remaining cases are similar. 
\end{proof}
%By lemma \ref{prop:splitting2} we can 
We write
$$
\lambda_1|_{\partial_{S,T}\Conf(\Gamma)}=\wt\pi_1^*\,\wt\lambda(\Gamma_{S,T}) \ \text{ and } \ \lambda_2|_{\partial_{S,T}\Conf(\Gamma)}=\wt\pi_2^*\,\lambda(\delta_{S,T}\Gamma).
$$
\begin{corollary}\label{lm:dim1}
Given $\Gamma\in \G$, $S\subset V$ and $T\subset W_\alpha$, let $\eta_{S,T}(\Gamma)$ be the number of edges connecting a vertex in $S\cup T$ with a vertex in $(V\cup W)\setminus (S\cup T)$. The contribution to equation \eqref{eq:external_sum} from the boundary stratum $\partial_{S, T}\Conf(\Gamma)$ vanishes unless $\eta_{S, T}(\Gamma)=2$.  
%The contribution to equation \eqref{eq:boundary} from the boundary strata corresponding to the vertices labeled by $S$ and $T$ coming together vanishes unless $\eta(\Gamma)=2$.    
\end{corollary}
\begin{proof}
By counting the number of edges connecting vertices in $S\cup T$ one finds 
$
\deg \wt\lambda(\Gamma_{S,T}) = 3|S|+|T|-\eta_{S,T}(\Gamma)
$. On the other hand, $\dim C_{S,T}= 3|S|+|T|-2$, and hence $\wt\lambda(\Gamma_{S,T})$ vanishes unless $\eta_{S,T}(\Gamma)\geq 2$. By a similar argument $\lambda(\delta_{S,T}\Gamma)$ vanishes on the boundary stratum unless ${\eta_{S,T}(\Gamma)\leq 2}$. 
\end{proof}
The following lemma is known as the STU relations.
\begin{lemma}\label{thm:STU} The contribution to equation \eqref{eq:external_sum} corresponding to two vertices coming together where at least one is external vanishes.
\end{lemma}
\begin{proof} Let $\Gamma_0$ be a graph with a single two-valent external vertex $v_0$ that has an incoming and an outgoing edge, and with all other vertices three- and one-valent. There are exactly three graphs $\Gamma_1,\Gamma_2,\Gamma_3\in \G$ that maps to $\Gamma_0$ upon collapsing two vertices. These graphs are shown in figure \ref{fig:STU}.%, where we imagine that the vertices and edges outside the encircled area are fixed for all three graphs. 
%We assume let $\Gamma_1,\Gamma_2\in G_{\mathbf{n},m}$ and $\Gamma_3\in \G_{\mathbf{n}',m+1}$ where $\mathbf n'=(n_1,n_2-1,n_3)$.
\begin{figure}[H]
\centering
\begin{tikzpicture}[vertex/.style={draw,circle, fill=black, inner sep=1.5pt}]
  \node (v1) at (-1.4,0)  {}; 
  \node (v2) at (1.4,0)  {};
  \node[vertex] (w) at (-0.6,0) {};
  \node[vertex] (v) at (0.6,0) {};
  \node (w1) at (-0.8,1.7)  {}; 
  \node (w2) at (0.8,1.7)  {};
  \node at (0,-1.2) {$\Gamma_1$};
  \draw[ultra thick,->] (v1) --  (v2);
  \draw[->] (w1)node[above]{$a$} -- (w)node[below]{$w_j$};
  \draw[<-] (w2)node[above]{$b$} -- (v)node[below]{$w_{j+1}$};
    \draw[dotted,thick] (1.1,0.4) arc (1.1:360:1.1);
\end{tikzpicture}\hspace{0.4cm}
\begin{tikzpicture}[vertex/.style={draw,circle, fill=black, inner sep=1.5pt}]
  \node (v1) at (-1.4,0)  {}; 
  \node (v2) at (1.4,0)  {};
  \node[vertex] (w) at (-0.6,0) {};
  \node[vertex] (v) at (0.6,0) {};
  \node (w1) at (-0.8,1.7)  {}; 
  \node (w2) at (0.8,1.7)  {};
  \node at (0,-1.2) {$\Gamma_2$};
  \draw[ultra thick,->] (v1) --  (v2);
  \draw[->] (w1)node[above]{$a$} -- (v)node[below]{$q_{j+1}$};
  \draw[<-] (w2)node[above]{$b$} -- (w)node[below]{$q_{j}$};
    \draw[dotted,thick] (1.1,0.4) arc (1.1:360:1.1);
\end{tikzpicture}\hspace{0.4cm}
\begin{tikzpicture}[vertex/.style={draw,circle, fill=black, inner sep=1.5pt}]
  \node (v1) at (-1.4,0)  {}; 
  \node (v2) at (1.4,0)  {};
  \node[vertex] (w) at (0,0) {};
  \node[vertex] (v) at (0,1) {};
  \node (w1) at (-0.8,1.7)  {}; 
  \node (w2) at (0.8,1.7)  {};
  \node at (0,-1.2) {$\Gamma_3$};
  \draw[ultra thick,->] (v1) --  (v2);
  \draw[<-] (w)node[below]{$q_j$} -- (v)node[left]{$p_i$};
  \draw[<-] (v) -- (w1)node[above]{$a$};
  \draw[->] (v) -- (w2)node[above]{$b$};
    \draw[dotted,thick] (1.1,0.4) arc (1.1:360:1.1);
\end{tikzpicture} 
\begin{tikzpicture}
  \node (a1) at (0,0)  {};
  \node (a2) at (0,2)  {$\longrightarrow$};
\end{tikzpicture}
\begin{tikzpicture}[vertex/.style={draw,circle, fill=black, inner sep=1.5pt}]
  \node (v1) at (-1.4,0)  {}; 
  \node (v2) at (1.4,0)  {};
  \node[vertex] (w) at (0,0) {};
  %\node[vertex] (v) at (0,1) {};
  \node (w1) at (-0.8,1.7)  {}; 
  \node (w2) at (0.8,1.7)  {};
  \node at (0,-1.2) {$\Gamma_0$};
  \draw[ultra thick,->] (v1) --  (v2);
  %\draw[<-] (w) -- (v)node[left]{$p_\gamma$};
  \draw[->] (w)node[below]{$q_0$} -- (w1)node[above]{$a$};
  \draw[<-] (w) -- (w2)node[above]{$b$};
    \draw[dotted,thick] (1.1,0.4) arc (1.1:360:1.1);
\end{tikzpicture} 
\caption{:~The STU-relation}
    \label{fig:STU}
\end{figure}
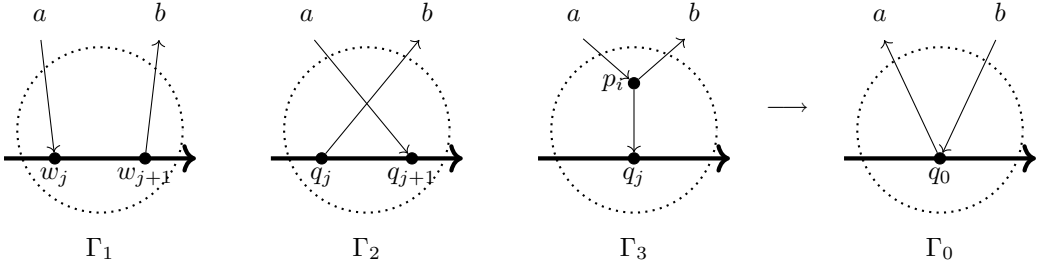
\noindent Collapsing the vertices $q_j$ and $q_{j+1}$ in $\Gamma_1$ and $\Gamma_2$ and the vertices $q_j$ and $p_i$ in $\Gamma_3$ into a single vertex $q_0$ we obtain a graph $\Gamma_0$ with a single two-valent external vertex as shown on the right-hands side of figure \ref{fig:STU}. The corresponding boundary strata are given by 
%Recall from section \ref{sec:ext_strata} that the boundary stratum of $\Conf_{\mathbf{n},m}$ coming from collapsing $q_\alpha$ and $q_{\alpha+1}$ is given by:
%$$
%\partial_{\{w_\alpha,w_\beta\},\emptyset}\Conf_{\mathbf{n},m}=\{*\}\times \Conf_{\mathbf{n}',m}
%$$
%where, as above, $\mathbf{n}'=(n_1,n_2-1,n_3)$. 
$$
\partial_{\,\emptyset,\{w_{j},w_{j+1}\}}\Conf(\Gamma_{k})=\{*\}\times\Conf(\Gamma_0)
$$
for $k=1,2$ and
$$
\partial_{\{v_i\},\{w_j\}}\Conf(\Gamma_3)=C_{\{v_i\},\{w_j\}}\times \Conf(\Gamma_0)\cong S^2\times \Conf(\Gamma_0).
$$
%By definition \ref{def:Conf_TS} it holds $C_{1,1}\cong S^2$, 
%and the restriction of $\lambda(\Gamma_3)$ to this boundary is given by
%$
%\wt\pi_1^*\omega\wedge \wt\pi_2^*(\lambda(\Gamma_0))
%$
%{\color{red}Assume wlog that the other ends of the edges are connected to vertices $p_\alpha$ and $p_\delta$.} 
We now determine the induced orientation on $\Conf(\Gamma_0)$ coming from each $\Gamma_k$, $k\in\{1,2,3\}$. By definition~\ref{def:Or} we can write
\begin{equation}\label{eq:Or_Gamma1}
\Or(\Gamma_1)=-\Or(\Gamma_2)=dq_j\wedge dq_{j+1}\wedge X
\end{equation}
and
\begin{equation}\label{eq:Or_Gamma2}
\Or(\Gamma_3)=(dq_j\wedge dp_i^1)\wedge dp_i^2\wedge dp_i^3\wedge X\,,
\end{equation}
where $X$ 
%the contribution from vertices outside of the encircled area, and 
is the same for all $\Gamma_k$, $k\in\{1,2,3\}$. Inserting $q_{j+1}=q_j+r$ for some $r>0$ into equation \eqref{eq:Or_Gamma1} we get
$$
\Or(\Gamma_1)= -\Or(\Gamma_2)= - dq_j\wedge dr \wedge X\,.
$$
Similarly, we can write $p_i=L_{t,\alpha}(q_j)+r u$ for some $r>0$ and unit vector $u\in \R^3$, and inserting this into \eqref{eq:Or_Gamma2} we get
\begin{align*}
\Or(\Gamma_3)= dq_j\wedge d^3(r u)\wedge  X=dq_j\wedge \vol_{S^2}\wedge\,r^2d r\wedge X\,.
\end{align*}
In each case, the vector $r$ is orthogonal to the boundary and pointing into the configuration space. Thus, fixing an orientation
$
\Or(\Gamma_0)=dq_j\wedge dX
$
on 
$
\Conf(\Gamma_0),
$
the contribution to equation \eqref{eq:external_sum} from the three boundary integrals takes the form:
\begin{equation*}
\int_{\Conf(\Gamma_0)}\lambda(\Gamma_0)\,c(\Gamma_0),
\end{equation*}
where $c(\Gamma_0)$ is the Lie algebra factor obtained from applying the usual Feynman rules to $\Gamma_0$ at each three- and one-valent vertex, and assigning to the two-valent vertex $q_0$ the factor:
\begin{align*}
{\zeta^a\xi_b}-{\xi_b\zeta^a}-{f^a}_{bc} \zeta^c\,\int_{S^2}\omega.
\end{align*}
Recall that from the definition in section \ref{sec:propagator} that $\omega$ integrates to one on $S^2$, and hence the above factor vanish by the Lie algebra relations:
\begin{equation*}
    [\zeta^a,\xi_b]={f^a}_{bc}\zeta^c\,.
\end{equation*}
Similar arguments would apply had we started from a graph $\Gamma_0$ with two incoming or two outgoing edges. Hence the theorem follows.
\end{proof}
\begin{lemma}\label{thm:hidden_ext}
Let $S\subset V$ and $T\subset W_\alpha$ such that $T\neq\emptyset$ and $|S\cup T|>2$. Then contribution to equation \eqref{eq:external_sum} from the boundary stratum where the vertices in $S\cup T$ come together vanishes. 
\end{lemma}
\begin{proof}
%Let $\Gamma_0$ be the graph obtained by collapsing the vertices in $S\cup T$ to a single external vertex $w_0$. Then it corresponding boundary stratum has the form
%$$
%\partial_{S,T}\Conf(\Gamma)=C_{S,T}\times \Conf(\Gamma_0).
%$$
%As in ... we write $\wt\lambda_1\in \Omega^*(C_{S,T})$ for the contribution to $\wt \lambda(\Gamma)$ coming from the edges connecting vertices in ...}. 
Recall from corollary \ref{lm:dim1} that we only get a contribution to equation \eqref{eq:external_sum} when $\Gamma$ has exactly two edges ``leaving the stratum'', that is, connecting a vertex in $S\cup T$ with a vertex not in $S\cup T$. We consider the following three cases separately: 
\begin{enumerate}[(a)]
\item Both of the edges leaving the stratum have orientations pointing out of $S\cup T$:
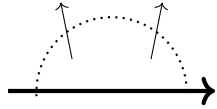
    \begin{center}
        \begin{tikzpicture}[vertex/.style={draw,circle, fill=black, inner sep=1.3pt}]
  \node (a1) at (-1.5,0)  {}; 
  \node (a2) at (1.5,0)  {};
  \node (v1) at (0.5,0.3) {};
  \node (v2) at (0.7,1.3) {};
  \node (w1) at (-0.5,0.3) {};
  \node (w2) at (-0.7,1.3) {};
  \draw[ultra thick,->] (a1) --  (a2);
  \draw[->] (v1) -- (v2);
  \draw[->] (w1) -- (w2);
    \draw[dotted, thick] (1,0) arc (1:183:1);
\end{tikzpicture} 
    \end{center}
    %\item Both edges connecting a vertex in $S\cup T$ with a vertex in $V(\Gamma)\setminus S\cup T$ are outgoing
    \item Both of the edges leaving the stratum have orientations pointing into $S\cup T$:
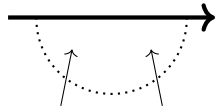
    \begin{center}
        \begin{tikzpicture}[vertex/.style={draw,circle, fill=black, inner sep=1.3pt}]
  \node (a1) at (-1.5,0)  {}; 
  \node (a2) at (1.5,0)  {};
  \node (v1) at (0.5,-0.3) {};
  \node (v2) at (0.7,-1.3) {};
  \node (w1) at (-0.5,-0.3) {};
  \node (w2) at (-0.7,-1.3) {};
  \draw[ultra thick,->] (a1) --  (a2);
  \draw[<-] (v1) -- (v2);
  \draw[<-] (w1) -- (w2);
    \draw[dotted, thick] (1,0) arc (1:-183:1);
\end{tikzpicture} 
    \end{center}
    \item One of the edges leaving the stratum has orientation point into $S\cup T$ and the other edge has orientation pointing out of $S\cup T$:
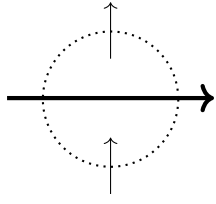
    \begin{center}
        \begin{tikzpicture}[vertex/.style={draw,circle, fill=black, inner sep=1.3pt}]
  \node (a1) at (-1.5,0)  {}; 
  \node (a2) at (1.5,0)  {};
  \node (v1) at (0,0.4) {};
  \node (v2) at (0,1.4) {};
  \node (w1) at (0,-0.4) {};
  \node (w2) at (0,-1.4) {};
  \draw[ultra thick,->] (a1) --  (a2);
  \draw[->] (v1) -- (v2);
  \draw[<-] (w1) -- (w2);
    \draw[dotted, thick] (0.9,0) arc (0.9:360:0.9);
\end{tikzpicture} 
    \end{center}
\end{enumerate}
%\textit{Case (a):} Since the only graphs are trees and there is no oriented path connecting a vertex on a Wilson line to a vertex on the same Wilson line.  we only have one type of internal vertex there must be two vertices coming together.
\paragraph{Case (a):} Since, by proposition \ref{prop:forest}, all contributing graphs are trees, this situation can only occur when $|S\cup T|=2$. 

\paragraph{Case (b):} We can assume that at least one of the edges leaving the stratum is connected to an internal vertex $v\in S$ since otherwise $|S|=\emptyset$ and $|T|=2$. Let $\Gamma_v$ be the disconnected sub-graph of $\Gamma_{S, T}$ spanned by the vertices $S'\cup T'$ connected by a path to $v$ %and let $\Gamma_w$ be the sub-graph spanned by vertices in $S\cup T$ connected by a path to $w$. 
as illustrated in figure \ref{fig:gamma_v}.
%and let $S'\subset S$ and $T'\subset T$ be the vertices in $\Gamma_v$. 
We write $$\wt\lambda(\Gamma_{S,T})=\wt\lambda_1\wedge \wt\lambda_2$$ where $\wt\lambda_1$ is constructed from edges in $\Gamma_v$ and $\wt\lambda_2$ is the contribution from the remaining edges in $\Gamma_{S,T}$. 
%can be written as a product:
%$$
%\wt\lambda_1=\wt\lambda_v\wedge \wt\lambda_w.
%$$ 
It then holds that $\wt \lambda_1$ factors through the projection
$$p: C_{S,T}\to C_{S',T'}$$ 
which forgets about the vertices not in $\Gamma_v$. By counting the number of edges and vertices in $\Gamma_v$ one finds that $\wt\lambda_1$ vanishes by the same dimensional arguments as used in the proof of corollary~\ref{lm:dim1}. 
\begin{figure}[H]
    \centering
\begin{tikzpicture}[vertex/.style={draw,circle, fill=black, inner sep=1.5pt}]
  \node (a1) at (-3,0)  {}; 
  \node (a2) at (3,0)  {};
  \draw[ultra thick,->] (a1) --  (a2);
  
  \node[vertex] (w) at (-0.5,0) {};
  \node[vertex] (v) at (0,-1) {};
  \node[vertex] (w1) at (-0.7,-1.7)  {}; 
  \node[vertex] (w2) at (.7,0)  {};
  \node[vertex] (w3) at (-1.5,0) {};
  \node[vertex] (v1) at (0.2,0) {};
  \node[vertex] (v2) at (1.5,0) {};
  \node[vertex] (v3) at (1,-1.4) {};
  \node (v4) at (1.2,-2.5) {};
  \node (w4) at (-1.2,-2.5) {};

  \draw[->] (v4) -- (v3)node[right]{};
  \draw[->] (v3) -- (v1);
  \draw[->] (v3) -- (v2);
  \draw[->] (w4) -- (w1);
  \draw[thick,red,<-] (w) -- (v);
  \draw[thick,red,->] (w1) -- (w3);
  \draw[thick,red,<-] (v) -- (w1)node[black,right]{$v$};
  \draw[thick,red,->] (v) -- (w2);
    \draw[dotted,thick] (2.3,0) arc (0:-180:2.3);
\end{tikzpicture} 
\caption{: Boundary stratum with two incoming edges. The red edges form the sub-graph $\Gamma_v$.}
\label{fig:gamma_v}
\end{figure}
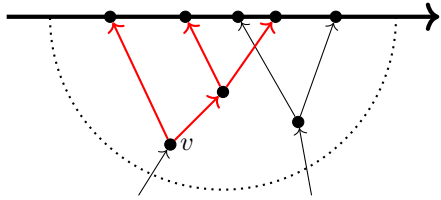    
\paragraph{Case (c):} We will further divide case (c) into two subcases:

\begin{enumerate}
    \item[(c1)] Either one of the edges leaving the stratum is connected an external vertex $w\in T$ or both edges leaving the stratum are connected to the same internal vertex $v\in S$. 

    \item[(c2)] Both edges leaving the stratum are connected to internal vertices $v,v'\in S$ and $v\neq v'$.

\end{enumerate}

\paragraph{Case (c1):} In this case $\wt \lambda(\Gamma_{S,T})$ vanishes on dimensional grounds by arguments completely analogous to those for case~(b). 

\paragraph{Case (c2):} Assume that the outgoing edge is connected $v\in S$. By assumption $\Gamma$ has two edges connecting $v$ to two different vertices in $S\cup T$. The situation is illustrated below where we have assigned coordinates $x,y$ and $z$ to the three vertices. Notice that $x$ and $z$ may be coordinates along the Wilson line.
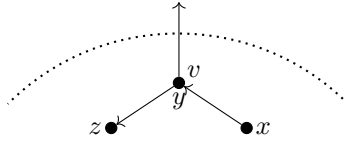
\begin{center}
        \begin{tikzpicture}[vertex/.style={draw,circle, fill=black, inner sep=1.5pt}]
  \node[vertex] (v0) at (0,0)  {}; 
  \node[vertex] (v1) at (0.9,-0.6)  {};
  \node[vertex] (v2) at (-0.9,-0.6) {};
  \node (v3) at (0,1.2) {};
  \draw[<-] (v0) --  (v1)node[right,black]{$x$};  
  \draw[->] (v0)node[below]{$y$} --  (v2)node[left,black]{$z$} ;
  \draw[->] (v0) --  (v3);
  \draw[dotted,thick] (2.25,-0.28) arc (45:135:3.2);

  \node at (0.2,0.15) {$v$};
  %\draw (v2) -- (-1.1,-0.9);
  %\draw (v2) -- (-0.7,-0.9);
  %\draw (v1) -- (1.1,-0.9);
  %\draw (v1) -- (0.7,-0.9);
\end{tikzpicture} 
\end{center}
We can now use a well known coordinate change originally due to Kontsevich \cite{K} to show the vanishing of the integral 
\begin{equation}\label{eq:Kont}
    \int_{C_{S,T}}\wt\lambda(\Gamma_{S,T}).
\end{equation}
In fact, integrating over $y$ in equation \eqref{eq:Kont} while keeping all other vertices fixed produces the integral
    \begin{equation}\label{eq:Kontsevich_int}
    \int_{y\in \R^3}\phi^*\omega(x,y)\wedge \phi^*\omega(y,z).
     \end{equation}
    We now make the following change of coordinates: $y = x+z-y'$.
    \begin{equation}
    \begin{aligned}
    \int_{y}\phi^*\omega(x,y)\wedge \phi^*\omega(y,z) &= - \int_{y'}\phi^*\omega(x,x+z-y')\wedge \phi^*\omega(x+z-y',z)\\&=-\int_{y'}\phi^*\omega(y',z)\wedge \phi^*\omega(x,y').
    \end{aligned}
    \end{equation}
    The minus since comes from this coordinate change being orientation reversing and the last equality uses translation invariance of $\phi$. This implies that the integral in equation \eqref{eq:Kontsevich_int} equals minus itself and hence must be zero.
\end{proof}
Lemma \ref{thm:STU} and \ref{thm:hidden_ext} proves theorem \ref{thm:ext}, and together with theorem \ref{thm:int} this completes the proof of theorem \ref{thm:1}. By lemma \ref{lm:infinity}, this implies that the expectation value $\mathcal R$ of a pair of crossing Wilson lines is a solution to the Yang-Baxter equation. In the following section we argue that $\mathcal R$ is in fact an $R$-matrix in the sence of section \ref{sec:YBE}. In particular, we show that $\mathcal R$ is independent of the angle of crossing between the Wilson lines and that it satisfies a so called unitarity relation, implying that it is invertible.

\section{Angle Independence and Unitarity}\label{sec:conclusion}
%Theorem \ref{thm:int} and \ref{thm:ext} prove theorem \ref{thm:1}, namely that $\braket{L_1}-\braket{L_0}=0$. By lemma \ref{lm:infinity}, this implies that the expectation value of a pair of crossing Wilson lines is a solution to the Yang-Baxter equation. However, as mention in remark \ref{rmk:R} of section \ref{sec:W-lines}, it is not immediately clear that the solution we obtain does not depend on a parameter corresponding to the angle of crossing between the lines. 
\begin{proposition}\label{prop:angle}
    Let $L$ and $L'$ be two (non-parallel) lines in $\R^2\times I$ supported at different points in $I$. Then expectation value $\mathcal R=\braket{LL'}$ is independent of the angle of crossing between the lines.
\end{proposition}
\begin{figure}[H]
    \centering
\begin{tikzpicture}[vertex/.style={draw,circle, fill=black, inner sep=1.2pt}]

\node (b1) at (-1.5,0)  {}; 
  \node (b2) at (-0.1,0)  {};
  \node (b3) at (0.1,0)  {};
  \node (b4) at (1.5,0)  {};
  \draw[very thick] (b1)node[left]{$L$} -- (b2);
  \draw[very thick,->] (b3) -- (b4);
  
  \node (c1) at (0,1.5)  {}; 
  \node (c2) at (0,-1.5)  {};
  \draw[very thick,->] (c1)node[above]{$L'$} --  (c2);
  \draw[->] (0.8,0) arc (0.8:-80:0.8);
  \node at (0.4,-0.4) {$\theta$};
\end{tikzpicture}
    \caption{}
    \label{fig:angle}
\end{figure}
\begin{proof}
Consider changing the angle $\theta$ at the crossing in figure \ref{fig:angle} by keeping $L'$ fixed while rotating $L$. We can apply the same vanishing arguments as in section \ref{sec:vanishing} to check that the expectation value is unchanged under this operation. Notice that in this case the tangent vector to $L$ dependents on $\theta$. 
%This means that $\lambda_1$ does not factor through the map $\wt\pi_1$ in lemma \ref{prop:splitting2} for points coming together on the bottom line.
We therefore get the following weaker version of corollary \ref{lm:dim1}: Let $\Gamma\in \G$ and let $S$ be a subset of internal vertices and $T$ a subset of external vertices on $L$. It then holds that $\lambda(\Gamma)$ vanishes on $\partial_{S,T}\Conf(\Gamma)$ unless $\eta_{S,T}(\Gamma)\leq 2$. On the other hand, since by proposition \ref{prop:forest} the only contributing Feynman graphs in are forests with roots on $L$ and leafs on $L'$, it holds that $\eta_{S,T}(\Gamma)\geq 2$ for any choice of $\Gamma$, and hence the vanishing arguments carry through regardless. 
\end{proof}
\begin{proposition}
    The element $\mathcal R$ is invertible, that is, it satisfies the relation shown in figure~\ref{fig:unitarity}.
\end{proposition}
\begin{figure}[H]
    \centering
\begin{tikzpicture}
  \draw[very thick] (0,0) arc (-90:90:1);
  \draw[very thick] (-0.4,2) -- (-1.3,2);
  \draw[very thick, ->] (-0.4,0) -- (-1.3,0);
  
  %\draw[very thick] (-0.2,3) -- (-0.2,2.2);
  %\draw[very thick] (-0.2,1.8) -- (-0.2,0.2);
  \draw[very thick, ->] (-0.2,3) -- (-0.2,-1);

  \node at (-0.5,2.3) {$R$};
  \node at (-0.5,-0.3) {$\wt R$};

  \node at (-0.2,-1.7) {(a)};
\end{tikzpicture} \hspace{4cm}
\begin{tikzpicture}

  \draw[very thick] (0,0) arc (-90:90:1);
  \draw[very thick] (0,2) -- (-1.3,2);
  \draw[very thick, ->] (0,0) -- (-1.3,0);
  
  \draw[very thick,->] (1.6,3) -- (1.6,-1);
  \node at (0.2,-1.7) {(b)};
\end{tikzpicture}
    \caption{}
    \label{fig:unitarity}
\end{figure}

\begin{proof}
We here use the exact same arguments as for the angle independence of $\mathcal R$ in proposition \ref{prop:angle}.
In fact, if we 
%consider the pair of Wilson lines in figure \ref{fig:unitarity} (a) with the corresponding expectation value given by the product $\mathcal R\wt {\mathcal R}$. 
start from the diagram in figure \ref{fig:unitarity} (a) and keep the top line fixed while continuously moving the bottom line to the left we obtain diagram in figure \ref{fig:unitarity} (b). By the same argument as above, the expectation value is invariant under this operation.
\end{proof}
\section{Conclusion}
We have proved that the expectation value $\mathcal R=\braket{LL'}$ of a pair of crossing Wilson lines is an $R$-matrix. In \cite{AK} Kaufman and the present author showed that the leading order deformation of the co-product in $\mathcal U_\hbar(\g)$ can be realised from the operation of merging two parallel Wilson lines. As in \cite{Aamand}, computations are here carried out in the setting of Chern-Simons theory for a semi-simple Lie algebra extended by an extra copy of the Cartan subalgebra. The arguments however translate directly into the present context. Together these results give a Wilson line realisation of the co-product and $R$-matrix in the quasi-triangular Hopf algebra $\mathcal U_\hbar (\g)$, thus supporting the claim that the category of Wilson line operators is equivalent to the category of representations of $\mathcal U_\hbar(\g)$ as a braided tensor category. 

A final remark worth noting: As mentioned in the introduction, the theory we have studied is equivalent to a topologically twisted $3$d $\mathcal N=4$ gauge theory. Moreover, if we take $\g=\mathfrak a\oplus \mathfrak a^*$ to be a Lie super-algebra this would also cover Chern-Simons theory as a $3$d $\mathcal N=4$ gauge theory with matter. We have here only considered the case when $\g$ is a classical Lie algebra but nothing in the arguments should change significantly if one instead considers the super-algebra case.

\section*{Acknowledgements}
I am grateful to Nathalie Wahl, Kevin Costello and Dani Kaufman for helpful discussions.\\
The author was supported by the European Research Council (ERC) under the European Union’s Horizon 2020 research and innovation programme (grant agreement No. 772960), and the Copenhagen Centre for Geometry and Topology (DNRF151)

\printbibliography

@article{BT,
  title={On the self‐linking of knots},
  author={Raoul Bott and Clifford H. Taubes},
  journal={Journal of Mathematical Physics},
  year={1994},
  volume={35},
  pages={5247-5287}
}

@article{BC,
  title={Integral invariants of 3-manifolds},
  author={Bott, Raoul and Cattaneo, Alberto S},
  journal={Raoul Bott: Collected Papers: Volume 5},
  pages={383},
  year={2018},
  publisher={Birkh{\"a}user}
}

@article{Aamand,
    author = "Aamand, Nanna Havn",
    title = "{Chern\textendash{}Simons theory and the R-matrix}",
    eprint = "1905.03263",
    archivePrefix = "arXiv",
    primaryClass = "hep-th",
    doi = "10.1007/s11005-021-01485-z",
    journal = "Lett. Math. Phys.",
    volume = "111",
    number = "6",
    pages = "146",
    year = "2021"
}

@article{AK,
    author = "Aamand, Nanna and Kaufman, Dani",
    title = "{A Wilson Line Realisation of Quantum Groups}",
    eprint = "2307.10830",
    archivePrefix = "arXiv",
    primaryClass = "hep-th",
    reportNumber = "CPH-GEOTOP-DNRF151",
    month = "7",
    year = "2023"
}

@article{AF,
    author={Altschuler, Daniel and Freidel, Laurent},
    title={Vassiliev knot invariants and Chern-Simons perturbation theory to all orders},
    journal={Comm. Math. Phys.},
    volume={187},
    pages={261--287},
    year={1997},
    publisher={Springer}
}

@inproceedings{K,
  title={Feynman diagrams and low-dimensional topology},
  author={Kontsevich, Maxim},
  booktitle={First European Congress of Mathematics Paris, July 6--10, 1992: Vol. II: Invited Lectures (Part 2)},
  pages={97--121},
  year={1994},
  organization={Springer}
}

@article{CWYI,
    author = "Costello, Kevin and Witten, Edward and Yamazaki, Masahito",
    title = "{Gauge Theory and Integrability, I}",
    eprint = "1709.09993",
    archivePrefix = "arXiv",
    primaryClass = "hep-th",
    reportNumber = "IPMU17-0136",
    doi = "10.4310/ICCM.2018.v6.n1.a6",
    journal = "ICCM Not.",
    volume = "06",
    number = "1",
    pages = "46--119",
    year = "2018"
}

@article{CWYII,
    author = "Costello, Kevin and Witten, Edward and Yamazaki, Masahito",
    title = "{Gauge Theory and Integrability, II}",
    eprint = "1802.01579",
    archivePrefix = "arXiv",
    primaryClass = "hep-th",
    reportNumber = "IPMU18-0025",
    doi = "10.4310/ICCM.2018.v6.n1.a7",
    journal = "ICCM Not.",
    volume = "06",
    number = "1",
    pages = "120--146",
    year = "2018"
}

@article{AS,
  title={Chern-Simons perturbation theory. II},
  author={Axelrod, Scott and Singer, Isadore M},
  journal={Journal of Differential Geometry},
  volume={39},
  number={1},
  pages={173--213},
  year={1994},
  publisher={Lehigh University}
}

@Article{FM,
 Author = {Fulton, William and MacPherson, Robert},
 Title = {A compactification of configuration spaces},
 FJournal = {Annals of Mathematics. Second Series},
 Journal = {Ann. Math. (2)},
 ISSN = {0003-486X},
 Volume = {139},
 Number = {1},
 Pages = {183--225},
 Year = {1994},
 Language = {English},
 DOI = {10.2307/2946631},
 zbMATH = {549326},
 Zbl = {0820.14037}
}

@article{Garner,
    author = "Garner, Niklas",
    title = "{Twisted Formalism for 3d $\mathcal{N}=4$ Theories}",
    eprint = "2204.02997",
    archivePrefix = "arXiv",
    primaryClass = "hep-th",
    month = "4",
    year = "2022"
}

@article{G2023,
  title={Vertex operator algebras and topologically twisted Chern-Simons-matter theories},
  author={Garner, Niklas},
  journal={Journal of High Energy Physics},
  volume={2023},
  number={8},
  pages={1--47},
  year={2023},
  publisher={Springer}
}

@article{CG,
    author = "Costello, Kevin and Gaiotto, Davide",
    title = "{Vertex Operator Algebras and 3d $ \mathcal{N} $ = 4 gauge theories}",
    eprint = "1804.06460",
    archivePrefix = "arXiv",
    primaryClass = "hep-th",
    doi = "10.1007/JHEP05(2019)018",
    journal = "JHEP",
    volume = "05",
    pages = "018",
    year = "2019"
}

@article{CMR,
  title={Perturbative quantum gauge theories on manifolds with boundary},
  author={Cattaneo, Alberto S and Mnev, Pavel and Reshetikhin, Nicolai},
  journal={Communications in Mathematical Physics},
  volume={357},
  pages={631--730},
  year={2018},
  publisher={Springer}
}

@article{CMK,
  title={Split Chern--Simons theory in the BV-BFV formalism},
  author={Cattaneo, Alberto S and Mnev, Pavel and Wernli, Konstantin},
  journal={Quantization, Geometry and Noncommutative Structures in Mathematics and Physics},
  pages={293--324},
  year={2017},
  publisher={Springer}
}

@article{Sinha,
  title={Manifold-theoretic compactifications of configuration spaces},
  author={Sinha, Dev P},
  journal={Selecta Mathematica},
  volume={10},
  pages={391--428},
  year={2004},
  publisher={Springer}
}

\end{document}